\documentclass[english]{article}
\usepackage[T1]{fontenc}
\usepackage[latin9]{inputenc}
\usepackage{geometry}
\usepackage{float}
\usepackage{booktabs}
\usepackage{mathtools}
\usepackage{amsmath}
\usepackage{amsthm}
\usepackage{amssymb}
\usepackage{graphicx}

\makeatletter

\providecommand{\tabularnewline}{\\}
\floatstyle{ruled}
\newfloat{algorithm}{tbp}{loa}
\providecommand{\algorithmname}{Algorithm}
\floatname{algorithm}{\protect\algorithmname}

\numberwithin{figure}{section}
\numberwithin{equation}{section}
\theoremstyle{remark}
\newtheorem*{rem*}{\protect\remarkname}
\theoremstyle{plain}
\newtheorem{thm}{\protect\theoremname}
\theoremstyle{plain}
\newtheorem{assumption}[thm]{\protect\assumptionname}
\theoremstyle{plain}
\newtheorem{lem}[thm]{\protect\lemmaname}

\usepackage{algorithm,algpseudocode}

\@ifundefined{showcaptionsetup}{}{%
 \PassOptionsToPackage{caption=false}{subfig}}
\usepackage{subfig}
\makeatother

\usepackage{babel}
\providecommand{\assumptionname}{Assumption}
\providecommand{\lemmaname}{Lemma}
\providecommand{\remarkname}{Remark}
\providecommand{\theoremname}{Theorem}

\begin{document}
\title{A Random Batch Method for Efficient Ensemble Forecasts of Multiscale
Turbulent Systems}
\author{Di Qi\textsuperscript{a} and Jian-Guo Liu\textsuperscript{b}}
\date{\textsuperscript{a }Department of Mathematics, Purdue University,
150 North University Street, West Lafayette, IN 47907, USA\\ \textsuperscript{b }Department of Mathematics and Department of Physics, Duke University, Durham, NC 27708, USA\\}
\maketitle
\begin{abstract}
A new efficient ensemble prediction strategy is developed for a general
turbulent model framework with emphasis on the nonlinear interactions
between large and small scale variables. The high computational cost
in running large ensemble simulations of high dimensional equations
is effectively avoided by adopting a random batch decomposition of
the wide spectrum of the fluctuation states which is a characteristic
feature of the multiscale turbulent systems. The time update of each
ensemble sample is then only subject to a small portion of the small-scale fluctuation
modes in one batch, while the true model dynamics with multiscale
coupling is respected by frequent random resampling of the batches
at each time updating step. We investigate both theoretical and numerical
properties of the proposed method. First, the convergence of statistical
errors in the random batch model approximation is shown rigorously
independent of the sample size and full dimension of the system. Then,
the forecast skill of the computational algorithm is tested on two
representative models of turbulent flows exhibiting many key statistical
phenomena with direct link to realistic turbulent systems. The random
batch method displays robust performance in capturing a series of
crucial statistical features of general interests including highly
non-Gaussian fat-tailed probability distributions and intermittent
bursts of instability, while requires a much lower computational cost
than the direct ensemble approach. The efficient random batch method
also facilitates the development of new strategies in uncertainty
quantification and data assimilation for a wide variety of complex
turbulent systems in science and engineering.
\end{abstract}

\section{Introduction and background}

Turbulent dynamical systems appearing in many natural and engineering
fields \cite{nicholson1983introduction,frisch1995turbulence,salmon1998lectures,majda2016introduction,tao2009multiscale}
are characterized by a wide range of spatiotemporal scales in a high
dimensional phase space. Small uncertainties in the multiscale high-dimensional
states can be rapidly amplified through the nonlinearly coupled dynamics
and inherent instability possessed by the turbulent flow. These distinctive
features give rise to a wide variety of complex phenomena such as
intermittent bursts of extreme flow structures and strongly non-Gaussian
probability density functions (PDFs) in the key state variables \cite{young2002srb,neelin2010long,majda2018strategies,wilcox1988multiscale}.
A probability representation for the evolution of the major flow states
is thus essential to accurately quantify the uncertainty in the practical
prediction of such turbulent systems. The ensemble forecast through
a Monte Carlo (MC) type approach estimates the evolution of the PDFs
by tracking an ensemble of trajectories solved independently from
an initial distribution \cite{leutbecher2008ensemble,robert1999monte,liu2017random}.
The empirical statistics of the ensemble solutions are used to approximate
the model uncertainty due to randomness from various internal and
external sources. A particular issue with large societal impacts is
to accurately capture the non-Gaussian PDFs related to the extreme
event outliers \cite{qi2022anomalous,gao2020transition,chen2020predicting}
using the finite size ensemble. However, the `curse-of-dimensionality'
forbids direct MC simulations of such high-dimensional systems especially
in cases including strongly coupled multiscale nonlinear interactions
\cite{robert1999monte,majda2016introduction}. A very large ensemble
size is usually needed to sufficiently sample the entire coupled fluctuation
modes in a wide energy spectrum, while only a small number of solutions
are affordable in many practical situations such as climate forecast
\cite{salmon1998lectures,gneiting2005weather}. It remains a grand
challenge to obtain accurate statistical estimates for the key physical
quantities from the multiscale interaction between the large-scale
mean flow and the interacting high-dimensional small-scale fluctuations. 

In this paper, we propose an efficient method for the understanding
and ensemble forecast of a general group of complex turbulent systems
accepting coupled multiscale dynamics \cite{majda2018strategies,majda2019linear}.
We use the ideas in the Random Batch Method (RBM) originally developed
for interacting particle systems \cite{jin2020random,jin2021convergence,gao2020data},
and apply it to the very different problem of multiscale turbulent
systems. Inspired by the stochastic gradient descent \cite{robbins1951stochastic,bubeck2015convex,hu2019diffusion}
in machine learning, the RBM randomly divides and constrains the large
number of interacting particles into small batches in each time interval.
The RBM can greatly reduce the computational cost in large particle
systems and has got many successful applications such as on manifold
learning \cite{gao2020transition,gao2020data} and quantum simulations
\cite{jin2020quantum}. In the ensemble prediction of turbulent systems,
we focus on the dominant flow structure in the largest scale thus
the required ensemble size to sample a low-dimensional subspace can
be controlled. This reduced modeling strategy usually suffers difficulties
in practice since the large scale is closely coupled with all the
unresolved small-scale fluctuations, thus it is impossible to only
perform ensemble simulation inside the large-scale subspace. Using
random batches, this difficulty is effectively overcome by regrouping
the large number of small-scale fluctuating modes into small batches
each containing only a few modes. Then the batches from a single simulation
of the fluctuation modes are used for updating different ensemble
members of the large-scale state separately during a short time step
update. This approximation is based on the important observation that
the small-scale fluctuations often decorrelate much faster in time
and contain less energy than the mean-flow state on the largest scale.
The batches of different modes are randomly resampled before each
time updating step thus contributions from all scales are well characterized
during the evolution in time. In this way, we achieve an efficient
algorithm that gains high skill to capture the fully non-Gaussian
statistical feature in the most important large-scale mean flow state,
while greatly reduce the high computational cost independent of the
dimensionality of the full system.

In order to achieve a detailed analysis of the method for the complex
turbulent system which usually combines various complex effects from
different sources, we develop the the new RBM algorithm on a unified
class of turbulent systems with emphasis on the explicit coupling
between large and small scales, while the unresolved nonlinear coupling
among small scales is modeled by damping and white noise forcing.
The simplified formulation reveals the most important key physical
processes on the nonlinear interaction between the large-scale mean
flow component and the smaller scale fluctuation components. On the
other hand, the extra complexity due to the nonlinear self-interactions
among the less important small scales are avoided to provide a cleaner
model setup. The model framework is shown to have many representative
applications in physical and engineering problems \cite{majda2006nonlinear,nicholson1983introduction,vallis2006atmospheric,olbers2001gallery,qi2020dimits}.
Precise error estimations are derived based on this general framework
for the convergence of the RBM approach using a finite number of samples.
The convergence of the statistical quantifies is proved based on the
semigroups generated by the backward Kolmogorov equations \cite{varadhan2007stochastic}
of the RBM model. The RBM is then verified numerically on two representative
prototype turbulent models. Different non-Gaussian statistics are
observed in the two models inferring strong intermittency and extreme
events induced from distict physical mechanisms. The numerical tests
show accurate prediction of both transient and equilibrium PDFs recovering
various non-Gaussian features under a much lower computational cost
requiring a much smaller ensemble size.

First in the following, we display the general structures that are
representative in the turbulent dynamics and illustrate the central
issues in achieving an accurate probability solution.

\vspace{-1em}

\subsection*{General formulation for turbulent systems and challenges in efficient
statistical forecast}

The complex turbulent systems discussed above can be written as the
following canonical equation about the state variable $\mathbf{u}$
in a high-dimensional phase space
\begin{equation}
\frac{\mathrm{d}\mathbf{u}}{\mathrm{d}t}=\Lambda\mathbf{u}+B\left(\mathbf{u},\mathbf{u}\right)+\mathbf{F}\left(t\right)+\boldsymbol{\sigma}\left(t\right)\dot{\mathbf{W}}\left(t;\omega\right).\label{eq:abs_formu}
\end{equation}
On the right hand side of the equation (\ref{eq:abs_formu}), the
first component, $\Lambda=-D+L$, represents linear dissipation and
dispersion effects (with a negative-definite dissipation operator
$D<0$ and a skew-symmetric dispersion operator $L^{T}=-L$ as in
\cite{majda2018strategies}). One representative feature of such complex
systems is the nonlinear energy conserving interaction that transports
energy across scales. The nonlinear effect is introduced through a
bilinear quadratic form, $B\left(\mathbf{u},\mathbf{u}\right)$, that
satisfies the conservation law $\mathbf{u}\cdot B\left(\mathbf{u},\mathbf{u}\right)\equiv0$.
External forcing effects are decomposed into a deterministic component,
$\mathbf{F}\left(t\right)$, and a stochastic component represented
by a Gaussian random process, $\boldsymbol{\sigma}\dot{\mathbf{W}}$. 

The evolution of the model state $\mathbf{u}$ depends on the sensitivity
to the randomness in initial conditions and external stochastic effects.
Combined with the inherent internal instability due to the nonlinear
coupling term in (\ref{eq:abs_formu}), small perturbations are amplified
in time thus requiring a probabilistic description to completely characterize
the development of uncertainty in the model state $\mathbf{u}$. The
time evolution of the PDF $p\left(\mathbf{u},t\right)$ can be found
directly from the solution of the associated Fokker-Planck equation
(FPE) \cite{varadhan2007stochastic}
\begin{equation}
\frac{\partial p}{\partial t}=\mathcal{L}_{\mathrm{FP}}p\coloneqq-\mathrm{div}_{\mathrm{u}}\left[\Lambda\mathbf{u}+B\left(\mathbf{u,u}\right)+\mathbf{F}\right]p+\frac{1}{2}\mathrm{div}_{\mathbf{u}}\nabla\left(\boldsymbol{\sigma}\boldsymbol{\sigma}^{T}p\right),\label{eq:FP}
\end{equation}
with an initial condition $p\mid_{t=0}=\mu_{0}$. However, it is still
a challenging task for directly solving the FPE (\ref{eq:FP}) as
a high dimensional PDE system. As an alternative approach, ensemble
forecast by tracking the Monte-Carlo solutions \cite{robert1999monte}
estimates the essential statistics through empirical averages among
a group of independently sampled trajectories of (\ref{eq:abs_formu}).
In particular, the ensemble members $\left\{ \mathbf{u}^{\left(i\right)}\right\} $
are sampled at the initial time $t=0$ according to the initial distribution
$\mu_{0}$. The PDF solution $p\left(\mathbf{u},t\right)$ at each
time instant $t>0$ is then approximated by evolving each sample independently
in time according to the same dynamical equation.

Though simple to implement, a direct ensemble forecast running the
original model (\ref{eq:abs_formu}) suffers several difficulties
in accurately recovering the key model statistics and PDFs in a high
dimensional space. First, the ensemble size required to achieve desirable
accuracy will grow exponentially in direct ensemble simulation of
the full model as the dimension of the system increases. Second, the
turbulent systems often contain strong internal instability and multiscale
coupling along the entire spectrum. Thus reduced models by simply
truncating the stabilizing small-scale modes \cite{majda2018strategies}
are not feasible to correctly represent the true model dynamics. Besides,
different orders of statistical characteristics are fully coupled
in the general formulation (\ref{eq:abs_formu}) so it is difficulty
to identify the contributions of different scales especially when
highly non-Gaussian statistics appear. These are the central difficulties
we will address in this paper using ideas in the RBM approximation.

In the following part of the paper, we first propose a more tractable
model framework with emphasis on the explicit coupling between large
and small scales in Section \ref{sec:A-turbulent-model}. Based on
this model framework, the RBM algorithm for ensemble prediction is
developed in Section \ref{sec:Random-batch-method}. Detailed error
estimation and convergence analysis of the new RBM method are obtained
in Section \ref{sec:Error-estimate}. The performance of the method
is verified on two prototype turbulence models with practical importance
in Section \ref{sec:Numerical-confirmation}. A summary of this paper
is given in Section \ref{sec:Summary}.

\section{A turbulent model framework with explicit multiscale coupling mechanism\label{sec:A-turbulent-model}}

We first introduce a simplified modeling framework that enables us
to focus on the key nonlinear large and small scale coupling mechanism
in the general system (\ref{eq:abs_formu}). One major difficulty
in complex turbulent systems is the strong nonlinear coupling across
scales where the large-scale state can destabilize the smaller scales
with small variance, while the increased fluctuation energy contained
in the large number of small-scale states can inversely impact the
development of the coherent largest scale structure. To address this
central issue of coupling with mixed scales in modeling turbulent
systems, we propose a mean-fluctuation decomposition of the model
state $\mathbf{u}$, so that the multiscale interactions can be identified
in a natural way. To achieve this, we view $\mathbf{u}$ as a random
field and separate it into the composition of a large-scale random
mean state and stochastic fluctuations in a finite-dimensional representation
under a pre-determined basis $\left\{ \mathbf{v}_{k}\right\} _{k=1}^{K}$
(for example, the Fourier expansion offers a natural basis for periodic
boundary)
\begin{equation}
\mathbf{u}\left(t;\omega\right)=\bar{\mathbf{u}}+\mathbf{u}^{\prime}\coloneqq\bar{\mathbf{u}}\left(t;\omega\right)+\frac{1}{\sqrt{K}}\sum_{k=1}^{K}Z_{k}\left(t;\omega\right)\mathbf{v}_{k},\label{eq:spec_expansion}
\end{equation}
where we use the overbar `$\overline{\bullet}$' to denote a proper
average operator. In this way, $\bar{\mathbf{u}}\left(t;\omega\right)$
represents the random mean field of the dominant largest scale structure
(for example, the zonal jets in geophysical turbulence \cite{majda2006nonlinear}
or the coherent radial flow in fusion plasmas \cite{diamond2005zonal});
and $\mathbf{Z}\left(t;\omega\right)=\left\{ Z_{k}\left(t;\omega\right)\right\} _{k=1}^{K}$
are stochastic coefficients measuring the uncertainty in multiscale
fluctuation processes $\mathbf{u}^{\prime}$ on the eigenmodes $\mathbf{v}_{k}$
(with a zero averaged mean $\overline{\mathbf{u}^{\prime}}=0$). Usually,
$\bar{\mathbf{u}}$ can be represented in a much lower dimension $d$
than the dimension $K$ of full stochastic modes $\mathbf{Z}$ representing
fluctuations along a wide spectrum of scales. The state decomposition
(\ref{eq:spec_expansion}) enables us to analyze the individual contributions
from different scale modes to the large and small scale dynamics.
Therefore, we derive the governing equations for the mean state and
fluctuation modes separately according to the decomposition (\ref{eq:spec_expansion}).

First, by averaging over the original equation (\ref{eq:abs_formu})
and applying the mean-fluctuation decomposition (\ref{eq:spec_expansion}),
the \emph{evolution equation of the large-scale mean state}\textbf{
$\bar{\mathbf{u}}$} is given by the following dynamics\addtocounter{equation}{0}\begin{subequations}\label{eq:dyn_multi}
\begin{equation}
\frac{\mathrm{d}\bar{\mathbf{u}}}{\mathrm{d}t}=\Lambda\bar{\mathbf{u}}+B\left(\bar{\mathbf{u}},\bar{\mathbf{u}}\right)+\frac{1}{K}\sum_{k,l=1}^{K}Z_{k}Z_{l}\bar{B}\left(\mathbf{v}_{k},\mathbf{v}_{l}\right)+\mathbf{F}.\label{eq:mean_dyn}
\end{equation}
Above, the small-scale nonlinear fluctuating feedback to the large-scale
mean dynamics is represented by the quadratic coupling $Z_{k}Z_{l}$
with the coupling coefficients $\bar{B}\left(\mathbf{v}_{k},\mathbf{v}_{l}\right)$.
The term $B\left(\bar{\mathbf{u}},\bar{\mathbf{u}}\right)$ represents
the self-interaction within the mean state. Next, by projecting the
fluctuation equation to each orthogonal basis element $\mathbf{v}_{i}$
we obtain the \emph{evolution equation for the stochastic fluctuation
coefficients}
\begin{equation}
\frac{\mathrm{d}Z_{k}}{\mathrm{d}t}=\frac{1}{K}\sum_{l=1}^{K}\gamma_{kl}\left(\bar{\mathbf{u}}\right)Z_{l}+\frac{1}{K^{3/2}}\sum_{m,n=1}^{K}Z_{m}Z_{n}B^{\prime}\left(\mathbf{v}_{m},\mathbf{v}_{n}\right)\cdot\mathbf{v}_{k}+\frac{\sigma\left(t\right)}{K^{1/2}}\dot{\mathbf{W}}\left(t;\omega\right)\cdot\mathbf{v}_{k},\label{eq:coeff_dyn}
\end{equation}
\end{subequations}where $\gamma_{kl}\left(\bar{\mathbf{u}}\right)=\left[\Lambda\mathbf{v}_{l}+B^{\prime}\left(\bar{\mathbf{u}},\mathbf{v}_{l}\right)+B^{\prime}\left(\mathbf{v}_{l},\bar{\mathbf{u}}\right)\right]\cdot\mathbf{v}_{k}$
characterizes the quasilinear coupling from the mean state $\bar{\mathbf{u}}$
in the fluctuation modes $\mathbf{Z}$. The interactions between the
fluctuation modes in different scales are summarized in the second
term on the right hand side of (\ref{eq:coeff_dyn}) with the fluctuation
coefficients $B^{\prime}=B-\bar{B}$. In addition, without loss of
generality, we assume that the deterministic forcing $\mathbf{F}$
exerts on the large-scale mean state, while the fluctuation modes
are subject to the coupled white noise forcing.

Still, the fully coupled mean-fluctuation model \eqref{eq:dyn_multi}
contains multiple linear and nonlinear interaction components involving
both large-scale mean $\bar{\mathbf{u}}$ and small-scale fluctuations
$\mathbf{u}^{\prime}$, thus it may not be a desirable starting model
for identifying the central dynamics in multiscale interactions. Rather,
we would like to propose a further simplified model based on mean-fluctuation
interaction mechanism which only maintains the key large and small
scale interaction explicitly while eliminates the large number of
small-scale self-interaction terms. Considering this, we assume that
the combined nonlinear feedback among different small-scale modes
(\ref{eq:coeff_dyn}) can be parameterized by independent damping
and stochastic forcing. This leads to the\emph{ simplified multiscale
model with large-small scale interaction}
\begin{equation}
\begin{aligned}\frac{\mathrm{d}\bar{\mathbf{u}}}{\mathrm{d}t}=\; & V\left(\bar{\mathbf{u}}\right)+\frac{1}{K}\sum_{k,l=1}^{K}Z_{k}Z_{l}\bar{B}\left(\mathbf{v}_{k},\mathbf{v}_{l}\right)+\mathbf{F},\\
\frac{\mathrm{d}Z_{k}}{\mathrm{d}t}= & \frac{1}{K}\sum_{l=1}^{K}\gamma_{kl}\left(\bar{\mathbf{u}}\right)Z_{l}\;-d_{k}Z_{k}+\sigma_{k}\dot{W}_{k},\qquad k=1,\cdots,K.
\end{aligned}
\label{eq:model_decoupled}
\end{equation}
Above, the linear and nonlinear effects within the mean state, $\Lambda\bar{\mathbf{u}}$
and $B\left(\bar{\mathbf{u}},\bar{\mathbf{u}}\right)$, are summarized
in a single term $V\left(\bar{\mathbf{u}}\right)$. The model simplification
occurs in the small-scale dynamics (\ref{eq:coeff_dyn}) for $Z_{k}$
which replaces the combined nonlinear small-scale coupling $Z_{m}Z_{n}B^{\prime}\left(\mathbf{v}_{m},\mathbf{v}_{n}\right)\cdot\mathbf{v}_{k}$
by equivalent damping and noise with parameters $d_{k},\sigma_{k}$.
In fact, this replaced term represents the higher-order moment feedback
to the covariance dynamics from a detailed statistical analysis of
the moment equations \cite{majda2018strategies}. The simplification
is derived from the important observation that these small-scale modes
are fast mixing (thus decorrelate fast in time) so the average of
the large number of modes plays an equivalent role as the linear damping
and white noise as in the homogenization theory. The introduced model
parameters $d_{k},\sigma_{k}$ are usually picked according to the
equilibrium statistical spectrum $E_{k}=\frac{\sigma_{k}^{2}}{2d_{k}}=E_{0}\left|k\right|^{-s}$
\cite{majda2006nonlinear} with $s$ determining the energy decaying
rate in the fluctuation modes. 

The resulting model (\ref{eq:model_decoupled}) recovers the most
essential coupling mechanism between the large-scale mean state $\bar{\mathbf{u}}$
and the small-scale fluctuation modes $Z_{k}$ which is explicitly
modeled through the quasi-linear operator $\gamma_{kl}\left(\bar{\mathbf{u}}\right)$
in the fluctuation equations and quadratic feedback term in the mean
equation. Notice that the strong internal instability which is the
key feature of turbulent systems is maintained in both the mean and
fluctuation modes by coupling terms $V$ and $\gamma_{kl}$ in (\ref{eq:mean_dyn})
and (\ref{eq:coeff_dyn}). The only model approximation comes from
parameterization of the complicated self-interactions of small-scale
fluctuations. Using this simplified model avoids the various sources
of uncertainties from the fluctuation scales so that the dominant
large-small scale interaction is identified. The thorough analysis
of the fully coupled nonlinearly fluctuation model (\ref{eq:coeff_dyn})
will be left for the future study.

In addition, the multiscale model formulation (\ref{eq:coeff_dyn})
enjoys the advantage of more flexibility to run ensemble simulations
for statistical forecast, uncertainty quantification and data assimilation
in practical applications. A wide variety of turbulent systems \cite{olbers2001gallery,majda2018model}
can be categorized into this framework so that it has wide validity
in developing the efficient ensemble forecast methods. Two typical
prototype models accepting the dynamical structure (\ref{eq:model_decoupled})
with a wide multiscale spectrum will be discussed in Section \ref{sec:Numerical-confirmation}
displaying a wide variety of different turbulent features. In the
following sections, we will develop efficient ensemble forecast strategies
based on this representative turbulent model formulation (\ref{eq:model_decoupled}). 

\section{Random batch method for ensemble forecast of turbulent models\label{sec:Random-batch-method}}

Next, we propose an efficient ensemble forecast method to describe
the time evolution of the probability distribution of the state $\mathbf{u}$.
The direct Monte-Carlo approach runs an ensemble simulation using
$N$ independent samples $\mathbf{u}^{\left(i\right)}=\left\{ \bar{\mathbf{u}}^{\left(i\right)},\mathbf{Z}^{\left(i\right)}\right\} ,i=1,\cdots,N$,
with $\bar{\mathbf{u}}\in\mathbb{R}^{d}$ the large-scale mean state
and $\mathbf{Z}=\left\{ Z_{k}\right\} _{k=1}^{K}\in\mathbb{R}^{K}$
the entire small-scale fluctuation modes in a high dimensional space
$K\gg d$. The samples are drawn from the initial distribution $\mathbf{u}^{\left(i\right)}\left(0\right)\sim\mu_{0}\left(\mathbf{u}\right)$
at the starting time $t=0$, and the time-dependent solution of each
sample $\mathbf{u}^{\left(i\right)}\left(t\right)$ is achieved by
solving the equation (\ref{eq:model_decoupled}) independently in
time. The resulting PDF at each time instant $t$ is approximated
by the empirical ensemble representation,
\begin{equation}
p\left(\mathbf{u},t\right)\simeq p^{\mathrm{MC}}\left(\mathbf{u},t\right)\coloneqq\frac{1}{N}\sum_{i=1}^{N}\delta\left(\mathbf{u}-\mathbf{u}^{\left(i\right)}\left(t\right)\right),\quad\mathbf{u}\in\mathbb{R}^{d+K}.\label{eq:pdf_empirical}
\end{equation}
Associated with the PDF, the statistical expectation of any function
$\varphi\left(\mathbf{u}\right)$ can be estimated by the empirical
average of the samples according to (\ref{eq:pdf_empirical})
\[
\mathbb{E}^{p}\varphi_{t}\left(\mathbf{u}\right)\simeq\mathbb{E}^{\mathrm{MC}}\varphi_{t}\left(\mathbf{u}\right)=\frac{1}{N}\sum_{i=1}^{N}\varphi\left(\mathbf{u}^{\left(i\right)}\left(t\right)\right).
\]
In particular for the model (\ref{eq:model_decoupled}), all the small-scale
modes contribute to the mean state equation as a combined feedback.
As a result, even though we are mostly interested in the statistics
from the mean state samples $\bar{\mathbf{u}}^{\left(i\right)}$ in
the relatively low dimensional subspace, the solution of entire $K$
small-scale modes $\mathbf{Z}^{\left(i\right)}$ must be computed.
The $K$-dimensional equations for fluctuation modes also need to
be solved repeatedly $N$ times for all the samples $i=1,\cdots,N$.
Thus, the direct ensemble method reaches a high computational cost
of $O\left(NK^{2}\left(d+1\right)\right)$ for one time step update.
Furthermore, the required number of samples $N$ to maintain accuracy
in the empirical PDF (\ref{eq:pdf_empirical}) is also dependent on
the system dimension $\left(d+K\right)$ and will grow exponentially
as $K$ increases (known as the curse-of-dimensionality \cite{daum2003curse,friedman1997bias}).
Therefore, this direct MC approach will quickly become computational
unaffordable as a larger $N$ is needed to resolve all the detailed
small scale fluctuations. 

Here, we propose to reduce the computational cost in ensemble simulations
of the turbulent models using the idea in the effective random batch
method (RBM) \cite{jin2020random,jin2021convergence}. We focus on
the ensemble sampling of the most dominant mean state statistics $\bar{\mathbf{u}}\in\mathbb{R}^{d}$
in a much lower dimensional subspace $d\ll K$. Thus accurate empirical
estimation of the marginal probability distribution of $\bar{\mathbf{u}}$
in (\ref{eq:model_decoupled}) can be reached using a much smaller
ensemble size $N_{1}\ll N$ 
\begin{equation}
p^{\mathrm{RBM}}\left(\bar{\mathbf{u}}\right)=\frac{1}{N_{1}}\sum_{i=1}^{N_{1}}\delta\left(\bar{\mathbf{u}}-\bar{\mathbf{u}}^{\left(i\right)}\right),\quad\bar{\mathbf{u}}\in\mathbb{R}^{d}.\label{eq:pdf_rbm}
\end{equation}
Accordingly, the expectation in the resolved mean state is computed
by the empirical estimation, $\mathbb{E}^{\mathrm{RBM}}\varphi_{t}\left(\bar{\mathbf{u}}\right)=\frac{1}{N_{1}}\sum_{i}\varphi\left(\bar{\mathbf{u}}^{\left(i\right)}\left(t\right)\right)$.
In the main idea of the RBM model, we no longer run the associated
large ensemble simulation of the full fluctuation modes $\left\{ \mathbf{Z}^{\left(i\right)}\right\} \in\mathbb{R}^{K\times N}$
associated with each mean state sample $\bar{\mathbf{u}}^{\left(i\right)}$.
Instead, only one stochastic trajectory of $\mathbf{Z}\left(t\right)$
is solved in time while the $K$ spectral modes are divided into smaller
subgroups (batches) for updating different ensemble members $\bar{\mathbf{u}}^{\left(i\right)}$.
The RBM approximation is introduced considering the typical property
of the turbulent system where the energy inside the single small-scale
mode $\mathbb{E}\left|Z_{k}\right|^{2},k\gg1$ decays fast and decorrelates
rapidly in time (see examples in Figure \ref{fig:Equilibrium-energy-spectrum}
and \ref{fig:Equilibrium-energy-topo} of Section \ref{sec:Numerical-confirmation}).
On the other hand, ergodicity in the stochastic fluctuation modes
\cite{weinan2001gibbsian,mattingly2002ergodicity} implies that updating
the mean state $\bar{\mathbf{u}}$ using fractional fluctuation modes
at each time step with consistent time-averaged feedback can provide
an equivalent total contribution without altering the original statistical
equations. 

Next, we describe the detailed RBM approach for modeling turbulent
systems. To accurately quantify the small-scale feedback in the mean
state dynamics, we introduce a partition $\mathcal{I}^{n}=\left\{ \mathcal{I}_{i}^{n}\right\} $
of the mode index $k=1,\cdots,K$ at the start of each time updating
step $t=t_{n-1}$. Thus, the full spectrum of modes is divided into
$N_{1}$ small batches of size $p=\left|\mathcal{I}_{i}^{n}\right|$
according to the total number of samples $i=1,\cdots,N_{1}$ of $\bar{\mathbf{u}}^{\left(i\right)}$
\begin{equation}
\cup_{i=1}^{N_{1}}\mathcal{I}_{i}^{n}=\left\{ k:1\leq k\leq K\right\} .\label{eq:partition}
\end{equation}
Then the $i$-th sample of the mean state $\bar{\mathbf{u}}^{\left(i\right)}$
is updated only using the fluctuation modes $\left\{ Z_{k}\right\} $
whose indices belong to the batch $k\in\mathcal{I}_{i}^{n}$ during
the time interval $t\in\left(t_{n-1},t_{n}\right]$ with time integration
step $\tau=t_{n}-t_{n-1}$\addtocounter{equation}{0}\begin{subequations}\label{eq:dyn_rbm}
\begin{equation}
\frac{\mathrm{d}\bar{\mathbf{u}}^{\left(i\right)}}{\mathrm{d}t}=V\left(\bar{\mathbf{u}}^{\left(i\right)}\right)+\sum_{k,l\in\mathcal{I}_{i}^{n}}c_{kl}Z_{k}Z_{l}\bar{B}\left(\mathbf{v}_{k},\mathbf{v}_{l}\right)+\mathbf{F}.\label{eq:dyn1_rbm}
\end{equation}
Correspondingly, only the modes $\left\{ Z_{k}\right\} $ whose indices
$k$ in the batch $\mathcal{I}_{i}^{n}$ are updated using the $i$-th
mean state solution
\begin{equation}
\frac{\mathrm{d}Z_{k}}{\mathrm{d}t}=\frac{1}{p}\sum_{l\in\mathcal{I}_{i}^{n}}\gamma_{kl}\left(\bar{\mathbf{u}}^{\left(i\right)}\right)Z_{l}\;-d_{k}Z_{k}+\sigma_{k}\dot{W}_{k},\quad k\in\mathcal{I}_{i}^{n}.\label{eq:dyn2_rbm}
\end{equation}
\end{subequations}Above, the coupled feedback of small-scale modes
in the batch is rescaled by the new combining coefficients $c_{kl}$
\begin{equation}
c_{kl}=\begin{cases}
\frac{1}{p}, & k=l,\\
\frac{1}{p}\frac{K-1}{p-1}, & k\ne l.
\end{cases}\label{eq:weight_rbm}
\end{equation}
that will be explained in Section \ref{sec:Error-estimate}. Equation
(\ref{eq:dyn1_rbm}) and (\ref{eq:dyn2_rbm}) combining all the batches
$i=1,\cdots,N_{1}$ give the complete formulation of the \emph{random
batch model for statistical ensemble forecast of turbulent systems}
during time interval $\left(t_{n-1},t_{n}\right]$. Then the batches
are regrouped again at the start of the next time step $t=t_{n}$
to repeat this procedure. Through the RBM reduction, only a very small
number of modes $p$ is needed in each batch $i$
\[
p=\frac{K}{N_{1}}\geq1.
\]
In practice, it is shown that the batch size $p$ can be sufficiently
small (for example, we take $p=5$ or even $p=2$ compared with $K=O\left(100\right)$
for the numerical tests in Section \ref{sec:Numerical-confirmation}).
Through the shared modes among batches, the computational cost of
the RBM model (\ref{eq:dyn_rbm}) is effectively reduced to $O\left(N_{1}dp^{2}\right)=O\left(dpK\right)$.
Especially, the cost won't have the exponential growth depending on
the full dimension $d+K$ of the system since the samples of size
$N_{1}$ are only used to capture statistics in the low dimensional
state $\bar{\mathbf{u}}\in\mathbb{R}^{d}$.

We summarize the random batch method for ensemble simulation of high
dimensional turbulent system in the following algorithm:

\begin{algorithm}
\begin{algorithmic}[1] 
\Require{At initial time $t=t_{0}$, draw random samples from the initial distribution $\bar{\mathbf{u}}^{\left(i\right)}\left(t_0\right)\sim\mu_{0}\left(\bar{\mathbf{u}}\right)$ for $i=1,\cdots,N_{1}$, and one associated sample for $\left\{ Z_{k}\left(t_0\right)\right\} _{k=1}^{K}$.}

\For{$n = 1$ while $n < T/\tau$, at the start of the time interval $t\in\left(t_{n-1},t_{n}\right]$  with time step $\tau=t_{n}-t_{n-1}$}
	\State{Get samples $\bar{\mathbf{u}}^{\left(i\right)}\left(t_{n-1}\right)$ and spectral modes $\left\{Z_{k}\left(t_{n-1}\right)\right\} _{k=1}^{K}$ from all output batches of previous time step;}
	\State{Partition the modes into $N_1$ batches $\mathcal{Z}_{i}=\left\{ Z_{k}\right\} _{k\in\mathcal{I}_{i}^{n}}$ with $\cup_{i=1}^{N_1}\mathcal{Z}_{i}=\left\{ Z_{k}\right\} _{k=1}^{K}$ according to \eqref{eq:partition};}
	\State{Update the samples $\bar{\mathbf{u}}^{\left(i\right)}\left(t_{n}\right)$ and the spectral modes $\left\{Z_{k}\left(t_{n}\right)\right\} _{k\in\mathcal{I}_{i}^{m}}$ in each batch $i$ to the next time step by solving the equations \eqref{eq:dyn_rbm}.}
\EndFor

\end{algorithmic}

\caption{RBM model for ensemble forecast of turbulent systems\label{alg:Sample-RBM}}
\end{algorithm}
\begin{rem*}
1. The total number of batches $N_{1}$ is associated with the number
of small-scale fluctuation modes $K$. In the RBM model \eqref{eq:dyn_rbm},
we assume that the number of fluctuation modes is large enough $K\gg1$
so that it can offer sufficient samples $N_{1}=K/p$ to approximate
the PDF of $\bar{\mathbf{u}}$. Sometimes, it is useful to expand
the sample size to $N_{1}=n_{2}K/p$ by running a small number $n_{2}$
of samples for the fluctuation modes $Z_{k}^{\left(j\right)}$ with
$j=1,\cdots,n_{2}$. The sampled solutions are divided into small
batches of size $p$ according to the index $j$ in the same way as
before. The batches can also change size $p_{i}=\left|\mathcal{I}_{i}^{n}\right|$
according to the energy in each mode.

2. In practical applications to estimate the marginal distribution
(\ref{eq:pdf_rbm}), we can even extend the number of the resolved
modes for the ensemble simulation to also include the leading (unstable)
fluctuation modes $Z_{k},k\leq K_{1}$. The numerical scheme can be
easily extended to this case following exactly the same steps of Algorithm
\ref{alg:Sample-RBM} and is shown in the explicit examples in Appendix
\ref{appen2:Detailed-RBM-formulation}.
\end{rem*}

\section{Error estimate of the random batch algorithm for ensemble prediction\label{sec:Error-estimate}}

In this section, we analyze the approximation error in the RBM model
in Algorithm \ref{alg:Sample-RBM} compared with the direct Monte-Carlo
approach of the turbulent model. Ensemble simulation of the coupled
large-small scale system (\ref{eq:model_decoupled}) is considered
for probabilistic prediction of the large-scale mean state $u\in\mathbb{R}^{d}$.
For the direct MC model, the governing equations for each ensemble
member of the full states $\left\{ u^{\left(i\right)},Z_{k}^{\left(i\right)}\right\} ,i=1,\cdots,N$
can be expressed as
\begin{equation}
\begin{aligned}\frac{\mathrm{d}u^{\left(i\right)}}{\mathrm{d}t}= & \;V\left(u^{\left(i\right)}\right)+\frac{1}{K}\sum_{k,l}Z_{k}^{\left(i\right)}Z_{l}^{\left(i\right)}B_{kl}+F,\\
\frac{\mathrm{d}Z_{k}^{\left(i\right)}}{\mathrm{d}t}= & \;\gamma_{k}\left(u^{\left(i\right)}\right)Z_{k}^{\left(i\right)}\;-d_{k}Z_{k}^{\left(i\right)}+\sigma_{k}\dot{W}_{k}^{\left(i\right)},\qquad1\leq k\leq K,
\end{aligned}
\label{eq:mc_analysis}
\end{equation}
with the multiscale coupling coefficient $B_{kl}=B\left(\mathbf{v}_{k},\mathbf{v}_{l}\right)$,
and self-interaction inside the mean state summarized in $V\left(u\right)$.
For simplicity, we adopt the diagonal coupling coefficients $\gamma_{k}$
in the fluctuation equation for $Z_{k}$ which can be achieved by
a proper change of basis from (\ref{eq:model_decoupled}). An ensemble
of samples $i=1,\cdots,N$ is drawn from the initial distribution
in the direct MC simulation. For each sample $i$, the mean state
$u^{\left(i\right)}$ is coupled with all the small-scale fluctuation
modes $\left\{ Z_{k}^{\left(i\right)}\right\} \in\mathbb{R}^{K}$
of the entire spectrum $k\leq K$ (with $K\gg1$ in a high dimensional
phase space). The samples are updated independently for each $i$
during the time evolution of solutions. The randomness in the model
(\ref{eq:mc_analysis}) comes from the small-scale white noise and
the uncertainty from the initial samples $u^{\left(i\right)}\left(0\right)\sim\mu_{0}$.

In contrast, the RBM model \eqref{eq:dyn_rbm} updates each mean state
sample $\tilde{u}^{\left(i\right)}$ using only a portion of small-scale
modes inside the corresponding batch $\mathcal{I}_{i}^{m}\subseteq\left\{ k:1\leq k\leq K\right\} $
at the time updating step $t_{m}$. For distinction in notations,
we use states $\left\{ \tilde{u},\tilde{Z}_{k}\right\} $ with `tildes'
to represent the RBM solutions compared with the full MC states $\left\{ u,Z_{k}\right\} $.
The dynamics of the $i$-th batch $\left\{ \tilde{u}^{\left(i\right)},\tilde{Z}_{k}\right\} _{k\in\mathcal{I}_{i}^{m}}$
during the time interval $t\in\left(t_{m-1},t_{m}\right]$ can be
formulated as
\begin{equation}
\begin{aligned}\frac{\mathrm{d}\tilde{u}^{\left(i\right)}}{\mathrm{d}t}= & \;V\left(\tilde{u}^{\left(i\right)}\right)+\sum_{k,l\in\mathcal{I}_{i}^{m}}c_{kl}\tilde{Z}_{k}\tilde{Z}_{l}B_{kl}+F,\\
\frac{\mathrm{d}\tilde{Z}_{k}}{\mathrm{d}t}= & \;\gamma_{k}\left(\tilde{u}^{\left(i\right)}\right)\tilde{Z}_{k}\;-d_{k}\tilde{Z}_{k}+\sigma_{k}\dot{\tilde{W}}_{k},\qquad k\in\mathcal{I}_{i}^{m},
\end{aligned}
\label{eq:rbm_analysis}
\end{equation}
with samples only drawn for the large-scale mean state $\tilde{u}^{\left(i\right)},i=1,\cdots,N$
(for simplicity, we take $N_{1}=N$ in Algorithm \ref{alg:Sample-RBM}).
Only one realization of the fluctuation modes $\tilde{Z}_{k},k\leq K$
is needed and the $K$ modes are shared among the $N$ mean state
samples $\tilde{u}^{\left(i\right)}$. At the start of each updating
time $t=t_{m-1}$, the modes $\left\{ \tilde{Z}_{k}\right\} $ are
randomly regrouped into new batches $\cup_{i}\mathcal{I}_{i}^{m}=\left\{ k:k\leq K\right\} $
of size $p\left(=\frac{K}{N}=O\left(1\right)\right)$. Only the modes
belonging to the $i$-th batch, $\tilde{Z}_{k},k\in\mathcal{I}_{i}^{m}$,
are updated using the $i$-th mean state $\tilde{u}^{\left(i\right)}$,
and only these modes in batch $i$ are used to update the mean state
sample $\tilde{u}^{\left(i\right)}$. Due to the batch approximation,
new coupling coefficients $c_{kl}$ in (\ref{eq:weight_rbm}) are
introduced. Notice that in this RBM model (\ref{eq:rbm_analysis}),
different samples $\tilde{u}^{\left(i\right)}$ are no longer independent
with each other during the entire time evolution since they are linked
by the batch resampling at the start of each time update. Therefore,
the RBM model has the additional randomness from the random partition
of modes $\left\{ \mathcal{I}^{m}\right\} _{m=1}^{\left[T/\tau\right]}$
with $\tau=t_{m}-t_{m-1}$ the discrete time step size. 

\subsection{Main theorem for the approximation error in the RBM ensemble model}

In the ensemble forecast, we are interested in recovering the statistics
through the empirical PDF (\ref{eq:pdf_empirical}) rather than each
individual trajectory solution. One effective way to calibrate the
statistical error is to compare the difference between the ensemble
averaged $\frac{1}{N}\sum_{i=1}^{N}\mathbb{E}\varphi\left(u^{\left(i\right)}\left(t_{n}\right)\right)$
and $\frac{1}{N}\sum_{i=1}^{N}\mathbb{E}\varphi\left(\tilde{u}^{\left(i\right)}\left(t_{n}\right)\right)$
from model (\ref{eq:mc_analysis}) and (\ref{eq:rbm_analysis}) for
any test function $\varphi\in C_{b}^{2}$. We propose the following
structural assumptions for the coupling parameters of the models (\ref{eq:mc_analysis})
and (\ref{eq:rbm_analysis}):
\begin{assumption}
\label{assu:assump_coeffs}In the mean state equation for $u$, suppose
that the bilinear coupling coefficients $B_{kl}$ in the mean state
dynamics are uniformly bounded\addtocounter{equation}{0}\begin{subequations}\label{eq:assump}
\begin{equation}
\max_{1\leq k,l\leq K}B_{kl}\leq C.\label{eq:assump1}
\end{equation}
And the self-coupling term $V\left(u\right)$ and its derivatives
up to the second order are uniformly bounded
\begin{equation}
\left\Vert V\right\Vert _{C^{2}}=\sum_{\left|\alpha\right|\leq2}\left\Vert \nabla^{\alpha}V\right\Vert _{\infty}\leq C.\label{eq:assump2}
\end{equation}
In the fluctuation equations for $\left\{ Z_{k}\right\} $, we assume
that there is no internal instability induced by $u$ in all small-scale
modes and the total noise amplitude is bounded. That is, there is
a positive constant $r$ independent of $u$, so that
\begin{equation}
\min_{1\leq k\leq K}\left\{ d_{k}-\gamma_{k}\left(u\right)\right\} \geq r>0,\quad\mathrm{and}\quad\sum_{k=1}^{K}\sigma_{k}^{2}\leq C.\label{eq:assump3}
\end{equation}
\end{subequations}
\end{assumption}

Assumption (\ref{eq:assump1}) is natural from the definition of the
quadratic bilinear form $B_{kl}=B\left(\mathbf{v}_{k},\mathbf{v}_{l}\right)$,
and assumption (\ref{eq:assump2}) for the self-coupling term $V\left(u\right)$
makes sure that this term does not have a rapid growth (this can be
guaranteed by imposing constraints on the maximum value of $u$).
Assumption (\ref{eq:assump3}) implies that the mean state induces
no internal instability to the fluctuation equations through the coupling
with the small-scale modes. This requires that we only partition the
stable small-scale modes into random batches for the time updating
in Algorithm \ref{alg:Sample-RBM}.

Under these assumptions, we have the following main theorem characterizing
the statistical error in the RBM ensemble prediction:
\begin{thm}
\label{thm:err_estimate}With Assumption \ref{assu:assump_coeffs}
satisfied, the empirical statistical estimation of the random batch
model (\ref{eq:pdf_rbm}) with discrete time step $\tau$ converges
to that of the full model solution (\ref{eq:pdf_empirical}) up to
final time $t=T$ as
\begin{equation}
\sup_{n\tau\leq T}\left|\frac{1}{N}\sum_{i=1}^{N}\mathbb{E}\varphi\left(\tilde{u}_{n}^{\left(i\right)}\right)-\mathbb{E}\varphi\left(u_{n}\right)\right|\leq C_{\varphi}\left(T\right)\tau,\label{eq:err_bnd}
\end{equation}
with any $\varphi\in C_{b}^{2}$, and $u_{n}=u\left(n\tau\right)$
by solving (\ref{eq:mc_analysis}), $\tilde{u}_{n}=\tilde{u}\left(n\tau\right)$
by solving (\ref{eq:rbm_analysis}). $C_{\varphi}$ is independent
of the sample size $N$ and the fluctuation modes dimension $K$.
\end{thm}

Above, (\ref{eq:err_bnd}) gives the error estimation of the averaged
RBM prediction of the statistics $\mathbb{E}\varphi\left(\tilde{u}_{n}^{\left(i\right)}\right)$
compared with the full ensemble method $\mathbb{E}\varphi\left(u_{n}^{\left(i\right)}\right)$.
Since the samples in the full ensemble model are independent and identical
under expectation, we have $\mathbb{E}\varphi\left(u_{n}\right)=\mathbb{E}\varphi\left(u_{n}^{\left(i\right)}\right)=\frac{1}{N}\sum_{i=1}^{N}\mathbb{E}\varphi\left(u_{n}^{\left(i\right)}\right)$.
Notice that this formula does not quantify the sufficient size of
the ensemble $N$ to reach accurate statistical prediction of $\mathbb{E}\varphi$.
In particular, the full ensemble model (\ref{eq:mc_analysis}) requires
sufficient sampling of the full phase space of high dimension $d+K\gg1$.
It ends up with an exponential growth in the sample size $N$ as the
dimension $K$ increases, thus suffers the curse-of-dimensionality.
In contrast, the RBM model (\ref{eq:rbm_analysis}) only needs to
sample the low-dimensional resolved subspace $d\ll d+K$. Thus the
required sample size $N$ for $\tilde{u}^{\left(i\right)}$ can be
controlled regardless of the high dimension $K$ of the full fluctuation
state.

\subsection{Proof of the main theorem}

In order to estimate the statistical error under the test function
$\varphi$, we first introduce the function $w_{Z}$ based on the
solution of the full model (\ref{eq:mc_analysis}) given the realization
of the small-scale stochastic process $Z\left(\cdot\right)$
\begin{equation}
w_{Z}\left(x,t\right)\coloneqq w\left(x,t\mid Z\right)=\frac{1}{N}\sum_{i=1}^{N}\varphi^{x}\left(u^{\left(i\right)}\left(t\right)\mid Z^{\left(i\right)}\left(s\right),s<t\right),\label{eq:test_func}
\end{equation}
Above, we use $x=\left\{ x_{i}\right\} _{i=1}^{N}\in\mathbb{R}^{d\times N}$
to denote the initial values of all the samples (with $u^{\left(i\right)}\left(0\right)=x_{i}$
for each initial state $x_{i}=\left(x_{i}^{1},\cdots,x_{i}^{d}\right)\in\mathbb{R}^{d}$).
The function $\varphi^{x}\in C_{b}^{2}\left(\mathbb{R}^{d}\right)$
evaluates the sample solution $u^{\left(i\right)}\left(t\right)$
starting at initial state $u^{\left(i\right)}\left(0\right)=x_{i}$.
The function $w_{Z}$ is defined according to the solutions of the
first equation of (\ref{eq:mc_analysis}) depending on the entire
time sequence of the stochastic process $Z\left(s\right),s<t$, with
$Z=\left\{ Z_{k}\right\} _{k=1}^{K}\in\mathbb{R}^{K}$. Next, we define
the deterministic function $w\left(x,t\right)$ after taking expectation
$\mathbb{E}^{Z}$ (we use superscript for expectation on the stochastic
process $Z\left(\cdot\right)$)
\begin{equation}
w\left(x,t\right)\coloneqq\mathbb{E}^{Z}w_{Z}\left(x,t\right)=\frac{1}{N}\sum_{i=1}^{N}\mathbb{E}_{x}\varphi\left(u^{\left(i\right)}\left(t\right)\right),\label{eq:test_func1}
\end{equation}
with $\mathbb{E}_{x}=\mathbb{E}\left(\cdot\mid u^{\left(i\right)}\left(0\right)=x_{i}\right)$
the expectation given sample initial values $x$. By definition, we
have the initial condition $w_{z}\left(x,0\right)=\frac{1}{N}\sum_{i}\varphi\left(x_{i}\right)=w\left(x,0\right)$.
The functions $w$ and $w_{z}$ satisfy the basic properties from
the semigroup \cite{oksendal2013stochastic,feng2017semi}. That is,
let $\mathcal{L}_{z}$ be the generator for the equation (\ref{eq:mc_analysis}),
$\partial_{t}w_{z}\left(x,t\right)=\mathcal{L}_{z}w_{z}\left(x,t\right)$
as in (\ref{eq:backward_full}). We have from the $L^{\infty}$ contraction
of the semigroup, $\left\Vert e^{t\mathcal{L}_{z}}\varphi\right\Vert _{\infty}\leq\left\Vert \varphi\right\Vert _{\infty}$,
so that,
\[
\left\Vert w\left(\cdot,t\right)\right\Vert _{\infty}=\left\Vert \mathbb{E}^{Z}e^{t\mathcal{L}_{Z}}w\left(\cdot,0\right)\right\Vert _{\infty}\leq\mathbb{E}^{Z}\left\Vert e^{t\mathcal{L}_{Z}}w\left(\cdot,0\right)\right\Vert _{\infty}\leq\mathbb{E}^{Z}\left\Vert w\left(\cdot,0\right)\right\Vert _{\infty}=\left\Vert w\left(\cdot,0\right)\right\Vert _{\infty}\leq C.
\]

Accordingly, we can define the associated function $\tilde{w}$ using
the solution $\tilde{u}^{\left(i\right)}\left(t\right)$ of the RBM
model (\ref{eq:rbm_analysis}) starting from $\tilde{u}_{0}^{\left(i\right)}=x_{i}$
as
\[
\tilde{w}\left(x,t\right)\coloneqq\frac{1}{N}\sum_{i=1}^{N}\mathbb{E}_{x}\left[\varphi\left(\tilde{u}^{\left(i\right)}\left(t\right)\right)\right].
\]
Here, the expectation $\mathbb{E}_{x}$ applies on the additional
randomness subject to random partition $\mathcal{I}^{m}$ of the batches
in $\left\{ Z_{k}\right\} $ at each time update as well as the stochastic
white noise. In order to characterize the statistical evolution of
the samples, we construct the discrete semigroup $\tilde{\mathcal{S}}$
associated with the RBM model in three steps according to Algorithm
\ref{alg:Sample-RBM} in the time interval $\left(t_{m-1},t_{m}\right]$.
First, starting with the initial function $\tilde{w}\left(x,t_{m-1}\right)$,
we fix the solution of the fluctuation modes during the next updating
interval, $Z\left(s\right),t_{m-1}<s\leq t_{m}$. Next, before the
time update at $t=t_{m-1}$, we partition the modes $Z=\left\{ Z_{k}\right\} $
into small batches $\mathcal{I}^{m}=\left\{ \mathcal{I}_{i}^{m}\right\} $.
Finally, we update the solution to the next time step $t=t_{m}$ by
integrating the backward equation (\ref{eq:backward_rbm}) of the
mean state with time step $\tau=t_{m}-t_{m-1}$. The one-step integration
of the RBM generator $\mathcal{\tilde{L}}_{z}^{\mathcal{I}^{m}}$
conditional on the partition $\mathcal{I}^{m}$ and the fluctuation
solution $Z$ during the $m$-th update time interval gives the conditional
updating operator
\[
\tilde{\mathcal{S}}_{Z}^{\mathcal{I}^{m}}\tilde{w}\left(x,t_{m-1}\right)\coloneqq e^{\tau\mathcal{\tilde{L}}_{Z}^{\mathcal{I}^{m}}}\tilde{w}\left(x,t_{m-1}\right).
\]
The one-step semigroup operator follows by taking expectation first
on the random partition $\mathbb{E}^{\mathcal{I}^{m}}$ then on the
stochastic process $\mathbb{E}^{Z}$ in the updating interval $Z\left(s\right),t_{m-1}<s\leq t_{m}$
\begin{equation}
\tilde{\mathcal{S}}\tilde{w}\left(x,t_{m-1}\right)\coloneqq\mathbb{E}^{Z}\mathbb{E}^{\mathcal{I}^{m}}\tilde{\mathcal{S}}_{Z}^{\mathcal{I}^{m}}\tilde{w}\left(x,t_{m-1}\right)=\mathbb{E}^{Z}\tilde{w}_{Z}\left(x,t_{m}\right)=\tilde{w}\left(x,t_{m}\right),\label{eq:semigroup}
\end{equation}
where we define $\tilde{w}_{Z}=\tilde{\mathcal{S}}_{Z}\tilde{w}=\mathbb{E}^{\mathcal{I}^{m}}\tilde{\mathcal{S}}_{Z}^{\mathcal{I}^{m}}\tilde{w}$,
and $\tilde{\mathcal{S}}=\mathbb{E}^{Z}\tilde{\mathcal{S}}_{Z}$.
Therefore, it can be shown that the semigroup is formed by
\[
\tilde{\mathcal{S}}^{\left(m\right)}\varphi\left(x\right)=\frac{1}{N}\sum_{i=1}^{N}\mathbb{E}_{x}\left[\varphi\left(\tilde{u}_{m}^{\left(i\right)}\right)\right]=\tilde{w}\left(x,t_{m}\right)=\tilde{\mathcal{S}}\tilde{w}\left(x,t_{m-1}\right)=\tilde{\mathcal{S}}\circ\tilde{\mathcal{S}}^{\left(m-1\right)}\varphi,
\]
with $\tilde{u}_{m}=\tilde{u}\left(t_{m}\right)$. 

Applying the backward Kolmogorov equation to the full MC model and
the reduced RBM model, we find the governing equations depending on
the realization $z\left(\cdot\right)$ \addtocounter{equation}{0}\begin{subequations}\label{eq:backward}
\begin{align}
\partial_{t}w_{z}=\mathcal{L}_{z}w_{z} & =\sum_{i=1}^{N}\left[V\left(x_{i}\right)+\frac{1}{K}\sum_{k,l}z_{k}z_{l}B_{kl}+F\right]\cdot\partial_{x_{i}}w_{z},\label{eq:backward_full}\\
\partial_{t}\tilde{w}_{z}=\mathcal{\tilde{L}}_{z}^{\mathcal{I}^{m}}\tilde{w}_{z} & =\sum_{i=1}^{N}\left[V\left(x_{i}\right)+\sum_{k,l}c_{kl}I_{i}^{m}\left(k\right)I_{i}^{m}\left(l\right)z_{k}z_{l}B_{kl}+F\right]\cdot\partial_{x_{i}}\tilde{w}_{z}.\label{eq:backward_rbm}
\end{align}
\end{subequations}Above, the RBM generator $\mathcal{\tilde{L}}_{z}^{\mathcal{I}^{m}}$
is also conditional on the random batch partition of modes $\mathcal{I}^{m}=\left\{ \mathcal{I}_{i}^{m}\right\} $
in each time interval $t\in\left(t_{m-1},t_{m}\right]$. To formulate
the equations under comparable terms, we introduce the index function
during the time interval $t_{m-1}<t\leq t_{m}$ as
\begin{equation}
I_{i}^{m}\left(k\right)=\begin{cases}
1, & \mathrm{if}\:k\in\mathcal{I}_{i}^{m},\\
0, & \mathrm{otherwise.}
\end{cases}\label{eq:index}
\end{equation}
The index function $I_{i}^{m}$ is fixed during the time interval
and will be modified at the beginning of each time updating step $t_{m}$
subject to the random partition (we neglect the subscript `$m$' for
time step $t_{m}$ in the rest part of this section for brevity of
notations).

The proof of Theorem \ref{thm:err_estimate} follows the method of
weak convergence in \cite{jin2021convergence}. The main difference
here is that we have to deal with the nonlinear coupling from the
fluctuation modes $Z$. We need the following lemmas to compare the
difference between the statistical functions $w$ and $\tilde{w}$.
The proofs of these lemmas can be found in Appendix \ref{appen1:Proofs-of-lemmas}.

First, Lemma \ref{lem:coeff_exp} indicates the values of the new
coupling coefficients $c_{kl}$ in (\ref{eq:weight_rbm}) of the RBM
model according to expectations on the random batch partition. 
\begin{lem}
\label{lem:coeff_exp}Let $I_{i}\left(k\right)$ be the index function
(\ref{eq:index}) indicating that the mode $k$ belongs to the batch
$\mathcal{I}_{i}$ with $k=1,\cdots,Np=K$ ($K$ the total number
of modes partitioned into $N$ random batches). Then we have the expectations
about the partition $\mathcal{I}$ for any $k$
\[
\mathbb{E}I_{i}^{2}\left(k\right)=\mathbb{E}I_{i}\left(k\right)=\frac{p}{K};
\]
and for any $k\neq l$
\[
\mathbb{E}I_{i}\left(k\right)I_{i}\left(l\right)=\frac{p}{K}\frac{p-1}{K-1}.
\]
\end{lem}

Next, the following lemma describes the estimates on the moments of
entire fluctuation modes. The supremum bound is guaranteed by the
crucial stability assumption (\ref{eq:assump3}) with uniform negative
damping in the dynamical equations of all the modes $Z_{k}$.
\begin{lem}
\label{lem:flucs_bound}Under the stable dynamics (\ref{eq:assump3})
, the fluctuation modes $Z_{k}$ and $\tilde{Z}_{k}$ in (\ref{eq:mc_analysis})
and (\ref{eq:rbm_analysis}) satisfy
\[
\sup_{t\leq T}\mathbb{E}\left\Vert Z_{t}\right\Vert ^{2q}<C_{q},\quad\sup_{t\leq T}\mathbb{E}\left\Vert \tilde{Z}_{t}\right\Vert ^{2q}<C_{q},
\]
for any integer $q\geq1$ and $\left\Vert Z_{t}\right\Vert ^{2}=\sum_{k=1}^{K}\left|Z_{k}\left(t\right)\right|^{2}$.
The constant $C_{q}$ is independent of $T$.
\end{lem}

The last lemma shows the regularity in the function $w\left(x,t\right)$
that is uniformly bounded up to second-order differentiations.
\begin{lem}
\label{lem:action_bound}For $t\leq T$, we have the uniform bound
for the characteristic function (\ref{eq:test_func1}) at least up
to second-order derivatives
\[
\left\Vert \nabla_{x}^{2}w\left(\cdot,t\right)\right\Vert _{\infty}=\sum_{i,j=1}^{N}\left\Vert \partial_{x_{i}x_{j}}^{2}w\left(\cdot,t\right)\right\Vert _{\infty}<C_{\varphi}\left(T\right),
\]
with the constant $C_{\varphi}\left(T\right)>0$ independent of $N$.
\end{lem}

With the above lemmas, we can give the proof of our main theorem.
First, we estimate the one-step error between $w_{z}\left(x,t_{m+1}\right)$
and $\tilde{\mathcal{S}}_{z}w\left(x,t\right)$, then the final estimate
follows by taking the expectation about $Z$. 
\begin{proof}
[Proof of the Theorem]In the RBM model, the batches of modes $\left\{ \tilde{Z}_{k}\right\} $
are regrouped before each time step update at $t_{m}=m\tau$. Thus
we first focus on the one step update during the time internal $t_{m}<t\leq t_{m+1}$
for all $m\leq n$ in the backward equations for the full model (\ref{eq:backward_full})
and the RBM model (\ref{eq:backward_rbm}). 

For the one-step time update of the RBM model and given $z\left(s\right),t_{m}<s\leq t_{m+1}$,
the discrete semigroup operator $\tilde{\mathcal{S}}_{z}^{\mathcal{I}}$
for a fixed batch partition $\mathcal{I}$ applied on the function
$w\left(x,t_{m}\right)$ according to (\ref{eq:backward_rbm}) has
the expanded expression
\[
\tilde{\mathcal{S}}_{z}^{\mathcal{I}}w\left(x,t_{m}\right)=e^{\tau\tilde{\mathcal{L}}_{z}^{\mathcal{I}}}w\left(x,t_{m}\right)=w\left(x,t_{m}\right)+\tau\tilde{\mathcal{L}}_{z}^{\mathcal{I}}w\left(x,t_{m}\right)+\int_{0}^{\tau}\left(\tau-s\right)\left(\tilde{\mathcal{L}}_{z}^{\mathcal{I}}\right)^{2}e^{s\tilde{\mathcal{L}}_{z}^{\mathcal{I}}}w\left(x,t_{m}\right)ds.
\]
Correspondingly for the full model (\ref{eq:backward_full}), we have
the continuous generator applying on the same function at $t=t_{m+1}$
\[
w_{z}\left(x,t_{m+1}\right)=e^{\tau\mathcal{L}_{z}}w\left(x,t_{m}\right)=w\left(x,t_{m}\right)+\tau\mathcal{L}_{z}w\left(x,t_{m}\right)+\int_{0}^{\tau}\left(\tau-s\right)\mathcal{L}_{z}^{2}e^{s\mathcal{L}_{z}}w\left(x,t_{m}\right)ds.
\]
Notice that above we consider the one-step update from $w\left(x,t_{m}\right)$
conditional on the realization of fluctuation modes $z$ during the
updating interval $\left(t_{m},t_{m+1}\right]$.

Next, we take the expectation on the random batch partition $\mathcal{I}$
at step $t_{m}$ on the RBM generator $\mathcal{\tilde{L}}_{z}^{\mathcal{I}}$.
Then the RBM generator becomes consistent with the full model generator
under the proper choice of the coefficients $c_{kl}$
\[
\mathbb{E}^{\mathcal{I}}\tilde{\mathcal{L}}_{z}^{\mathcal{I}}=\mathcal{L}_{z}.
\]
This is guaranteed by Lemma \ref{lem:coeff_exp} and the coefficients
$c_{kl}$ in (\ref{eq:weight_rbm}) appear naturally from the lemma.
With this, we have from the definition in (\ref{eq:semigroup})

\begin{align*}
\tilde{\mathcal{S}}w\left(x,t_{m}\right)-w\left(x,t_{m+1}\right) & =\mathbb{E}^{Z}\mathbb{E}^{\mathcal{I}}\tilde{\mathcal{S}}_{Z}^{\mathcal{I}}w\left(x,t_{m}\right)-\mathbb{E}^{Z}w_{Z}\left(x,t_{m+1}\right).\\
 & =\int_{0}^{\tau}\left(\tau-s\right)\left[\mathbb{E}^{Z,\mathcal{I}}\left(\tilde{\mathcal{L}}_{Z}^{\mathcal{I}}\right)^{2}e^{s\tilde{\mathcal{L}}_{Z}^{\mathcal{I}}}-\mathbb{E}^{Z}\mathcal{L}_{Z}^{2}e^{s\mathcal{L}_{Z}}\right]w\left(x,t_{m}\right)ds.
\end{align*}
Then, we show that the residual terms in the above integrants are
uniformly bounded with the constant $C$ independent of $N,K$
\[
\left\Vert \mathbb{E}^{Z}\mathcal{L}_{Z}^{2}e^{s\mathcal{L}_{Z}}w\left(\cdot,t\right)\right\Vert _{\infty}<C,\quad\left\Vert \mathbb{E}^{Z}\left(\tilde{\mathcal{L}}_{Z}^{\mathcal{I}}\right)^{2}e^{s\tilde{\mathcal{L}}_{Z}^{\mathcal{I}}}w\left(\cdot,t\right)\right\Vert _{\infty}<C.
\]
In fact, using the fact that $\mathcal{L}_{z}^{2}$ and $e^{s\mathcal{L}_{z}}$
are commutative and $e^{s\mathcal{L}_{z}}$ is a contraction under
$L^{\infty}$, we have the estimation 
\[
\begin{aligned}\left\Vert \mathbb{E}^{Z}\mathcal{L}_{Z}^{2}e^{s\mathcal{L}_{Z}}w\left(\cdot,t\right)\right\Vert _{\infty} & \leq\mathbb{E}^{Z}\left\Vert e^{s\mathcal{L}_{Z}}\mathcal{L}_{Z}^{2}w\left(\cdot,t\right)\right\Vert _{\infty}\leq\mathbb{E}^{Z}\left\Vert \mathcal{L}_{Z}^{2}w\left(\cdot,t\right)\right\Vert _{\infty}\\
 & \leq\left(C_{0}+\frac{C_{1}}{K}\sum_{k,l}\mathbb{E}\left|Z_{k}Z_{l}\right|+\frac{C_{2}}{K^{2}}\sum_{k,l,m,n}\mathbb{E}\left|Z_{k}Z_{l}Z_{m}Z_{n}\right|\right)\left\Vert \nabla_{x}^{2}w\left(\cdot,t\right)\right\Vert _{\infty}.
\end{aligned}
\]
In the last inequality above, we expand the operator $\mathcal{L}_{Z}^{2}$
and $C_{0},C_{1},C_{2}$ are three positive constants independent
of $N,K$, and $\left\Vert \nabla_{x}^{2}w\right\Vert _{\infty}=\sum_{i,j}\left\Vert \partial_{x_{i}x_{j}}^{2}w\right\Vert _{\infty}$.
The uniform bounds for the total moments $\frac{1}{K}\sum_{k,l}\mathbb{E}\left|Z_{k}Z_{l}\right|\leq\mathbb{E}\left\Vert Z\right\Vert ^{2}$
and $\frac{1}{K^{2}}\sum_{k,l,m,n}\mathbb{E}\left|Z_{k}Z_{l}Z_{m}Z_{n}\right|\leq\mathbb{E}\left\Vert Z\right\Vert ^{4}$
are guaranteed by Lemma \ref{lem:flucs_bound}. And we have the uniform
bound for $\left\Vert \nabla_{x}^{2}w\right\Vert _{\infty}$ by Lemma
\ref{lem:action_bound}. The same result can be achieved for $\tilde{\mathcal{L}}_{Z}^{\mathcal{I}}$
with the similar argument. 

Therefore, during the time updating from $t_{m-1}$ to $t_{m}$, we
have the one-step error between the RBM solution $\tilde{\mathcal{S}}w\left(x,t_{m}\right)$
and the full model $w\left(x,t_{m+1}\right)$ from the time integration
of the residual term inside the time interval
\[
\left\Vert \tilde{\mathcal{S}}w\left(\cdot,t_{m}\right)-w\left(\cdot,t_{m+1}\right)\right\Vert _{\infty}=\left\Vert \mathbb{E}^{Z}\left[\tilde{\mathcal{S}}_{Z}w\left(\cdot,t_{m}\right)-w_{Z}\left(\cdot,t_{m+1}\right)\right]\right\Vert _{\infty}\leq C\tau^{2}.
\]
Finally, applying $\tilde{\mathcal{S}}^{\left(n\right)}$ on the initial
function $\varphi$ and using $\tilde{\mathcal{S}}\tilde{w}\left(x,t_{m}\right)=\tilde{w}\left(x,t_{m+1}\right)$
by recurrently applying the semigroup and using the $L^{\infty}$
contraction property for $\tilde{\mathcal{S}}$, the total error at
$t=t_{n}=n\tau$ can be computed as
\begin{align*}
\left\Vert \tilde{\mathcal{S}}^{\left(n\right)}\varphi-w\left(\cdot,t_{n}\right)\right\Vert _{\infty} & \leq\left\Vert \tilde{\mathcal{S}}\left[\tilde{\mathcal{S}}^{\left(n-1\right)}\varphi-w\left(\cdot,t_{n-1}\right)\right]\right\Vert _{\infty}+\left\Vert \tilde{\mathcal{S}}w\left(\cdot,t_{n-1}\right)-w\left(\cdot,t_{n}\right)\right\Vert _{\infty}\\
 & \leq\left\Vert \tilde{\mathcal{S}}^{\left(n-1\right)}\varphi-w\left(\cdot,t_{n-1}\right)\right\Vert _{\infty}+\left\Vert \tilde{\mathcal{S}}w\left(\cdot,t_{n-1}\right)-w\left(\cdot,t_{n}\right)\right\Vert _{\infty}\\
 & \leq\sum_{m=0}^{n-1}\left\Vert \tilde{\mathcal{S}}w\left(\cdot,t_{m}\right)-w\left(\cdot,t_{m+1}\right)\right\Vert _{\infty}\leq C\left(t_{n}\right)\tau.
\end{align*}
This gives the final error estimate by maximizing among the initial
samples $x_{i}$.
\end{proof}

\section{Numerical tests of the random batch algorithm on turbulent models\label{sec:Numerical-confirmation}}

Now, we evaluate the numerical performance of the general RBM model
in Algorithm \ref{alg:Sample-RBM} using turbulent models containing
a wide spectrum of fluctuation modes. In particular, we test the algorithm
on two prototype benchmark models that are shown to generate various
representative phenomena in multiscale turbulence, that is, the conceptual
turbulent model and the topographic barotropic model. The RBM model
displays uniformly high skill in capturing the key statistical behaviors
in the dominant large-scale states displaying highly non-Gaussian
PDFs and intermittent extreme events with much lower computational
cost.

\subsection{Random batch algorithm for the conceptual turbulent model}

One particular concrete example accepting the general model framework
(\ref{eq:model_decoupled}) is the \emph{conceptual dynamical model
for turbulence} developed in \cite{majda2014conceptual}
\begin{equation}
\begin{aligned}\frac{\mathrm{d}\bar{u}}{\mathrm{d}t}= & -\bar{d}\bar{u}+\frac{\gamma}{K}\sum_{k=1}^{K}v_{k}^{2}-\bar{\alpha}\bar{u}^{3}+\bar{F},\\
\frac{\mathrm{d}v_{k}}{\mathrm{d}t}= & -\gamma\bar{u}v_{k}-d_{k}v_{k}+\sigma_{k}\dot{W}_{k},\quad1\leq k\leq K.
\end{aligned}
\label{eq:conceptual_model}
\end{equation}
Above, the state variables $\left(\bar{u},v_{k}\right)\in\mathbb{R}^{1+K}$
constitute a $\left(K+1\right)$-dimensional system. The scalar large-scale
mean state $\bar{u}$ is coupled with each small-scale mode $v_{k}$
through the nonlinear interaction coefficient $\gamma>0$, while the
large number of small-scale fluctuation modes $v_{k}$ impact the
large-scale mean state $\bar{u}$ together through a quadratic coupling
term. It is easy to check that the nonlinear coupling term conserves
the total energy $E=\frac{1}{2}\left(\bar{u}^{2}+\frac{1}{K}\sum_{k}v_{k}^{2}\right)$
so that the structural property in (\ref{eq:abs_formu}) is satisfied.
The model (\ref{eq:conceptual_model}) gives a typical characterization
for the anisotropic turbulence in which fluctuating energy flows intermittently
from a wide range of small scales to affect the largest scale mean
flow. Besides, unstable dynamics are also induced in the leading fluctuation
modes $v_{k}$ through coupling with the mean state when $-d_{k}-\gamma\bar{u}>0$,
where strong intermittency and extreme events are triggered with non-Gaussian
statistics through the chaotic fluctuations. This characterizes another
key observation in turbulence which is captured in this conceptual
model.

The model parameters $\left(\bar{d},\bar{\alpha},\bar{F},\gamma,d_{k},\sigma_{k}\right)$
for a strongly unstable regime are listed in Table \ref{tab:Model-parameters-conceptual}.
Strong instability in the large-scale dynamics for $\bar{u}$ is imposed
through the linear anti-damping term $\bar{d}<0$, which needs to
be balanced by the nonlinear feedbacks from both small and large scales.
A wide spectrum of modes $K=100$ is included for multiscale fluctuations.
$d_{k}$ and $\sigma_{k}$ are used to describe the turbulent dissipation
and white noise forcing in these fluctuation modes respectively. For
the convenience of the numerical test, we assign the Kolmogorov spectrum
$E_{k}=\frac{\sigma_{k}^{2}}{2d_{k}}=E_{0}k^{-5/3}$ for the leading
modes $k\leq K_{1}=4$, while all the other smaller-scale modes have
equipartition of energy $E_{k}=E_{0}K_{1}^{-5/3}$. Under this model
setup, it can be shown that an ergodic invariant measure \cite{mattingly2002ergodicity}
will be reached while instability on both small and large scales will
create intermittent behavior and non-Gaussian distributions during
the model evolution.

For direct ensemble simulation in a large phase space with dimension
$1+K=101$, we take a large ensemble size $N=1\times10^{4}$ for accurate
MC simulation to get the reference true distribution. Next, to apply
the RBM model to recover the probabilistic solutions in leading states,
only the mean state $\bar{u}$ together with the first 4 leading fluctuation
modes $v_{k},k\leq K_{1}=4$ is sampled. This is considering their
unstable dynamics to satisfy the assumption in (\ref{eq:assump3})
so that only stable small-scale fluctuation modes are partitioned
in the random batches of a small size $p=5$. Thus a much smaller
sample size $N_{1}=100$ is sufficient to model the $1+K_{1}=5$ dimensional
subspace. The algorithm can be developed directly according to the
steps in Algorithm \ref{alg:Sample-RBM} with a simple generalization
of also resolving the leading fluctuation modes. The joint PDFs of
the resolved states can be approximated by sample histograms as the
empirical representations in (\ref{eq:pdf_empirical}) and (\ref{eq:pdf_rbm}).
We summarize the detailed numerical RBM equations for the conceptual
model (\ref{eq:conceptual_model}) in Appendix \ref{subsec:RBM-conceptual}.

\begin{table}
\begin{centering}
\begin{tabular}{cccccccccccccc}
\toprule 
$K$ & $\bar{d}$ & $\bar{\alpha}$ & $\bar{F}$ & $\gamma$ &  & $d_{k}$ & $\sigma_{k}$ & $E_{0}$ &  & $K_{1}$ & $p$ & $N$ (full MC) & $N_{1}$ (RBM)\tabularnewline
\midrule 
100 & -0.1 & 0.05 & -0.055 & 1.5 &  & $1+0.02k^{2}$ & $\sqrt{2E_{k}d_{k}}$ & 0.004 &  & 4 & 5 & 10000 & 100\tabularnewline
\bottomrule
\end{tabular}
\par\end{centering}
\caption{Parameter values for the conceptual turbulent model.\label{tab:Model-parameters-conceptual}}
\end{table}
To illustrate the basic statistical features, Figure \ref{fig:Equilibrium-energy-spectrum}
plots the equilibrium energy spectrum $\mathbb{E}^{\mathrm{eq}}\left|v_{k}\right|^{2}$
and the decorrelation time $\int_{0}^{\infty}\mathbb{E}^{\mathrm{eq}}\left|v_{k}\left(t\right)v_{k}\left(0\right)\right|\mathrm{d}t$
of the fluctuation modes (with $\mathbb{E}^{\mathrm{eq}}$ denoting
the average about the equilibrium measure). This displays a typical
example of the common features in turbulent flows: the first few leading
modes accumulate most of the energy and a relatively long decorrelation
time, while there exists a long extended spectrum of small-scale fast-mixing
fluctuating modes containing small energy in each mode but having
a non-negligible combined contribution to the large-scale mean flow.
This is shown more clearly in the right column of Figure \ref{fig:Equilibrium-energy-spectrum}
for the time evolution of energy $\sum_{k}\left|v_{k}\right|^{2}$
in the leading modes $k\leq4$ and all the rest small-scale modes
$k>4$. The leading modes show intermittent bursts of large energy
due to the destabilizing coupling with the mean flow, while the large
number of small-scale fluctuating modes account for a major amount
of energy during the quiescent regime for most of the time during
the evolution. This confirms that the contributions from the many
small-scale modes play an important role of driving the large scales
to the final equilibrium and cannot be simply ignored in simulations.

\begin{figure}
\subfloat{\includegraphics[scale=0.27]{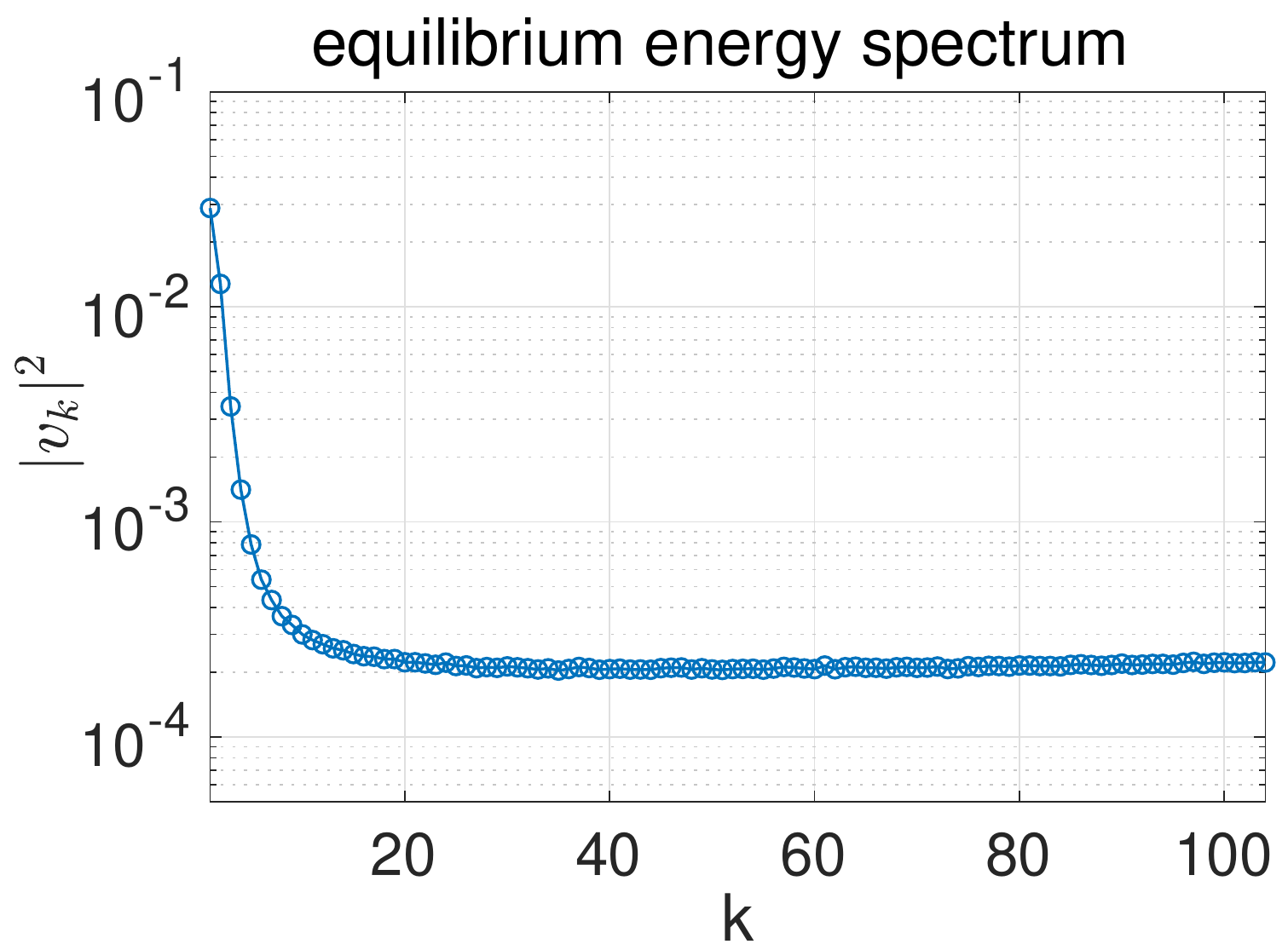}\includegraphics[scale=0.27]{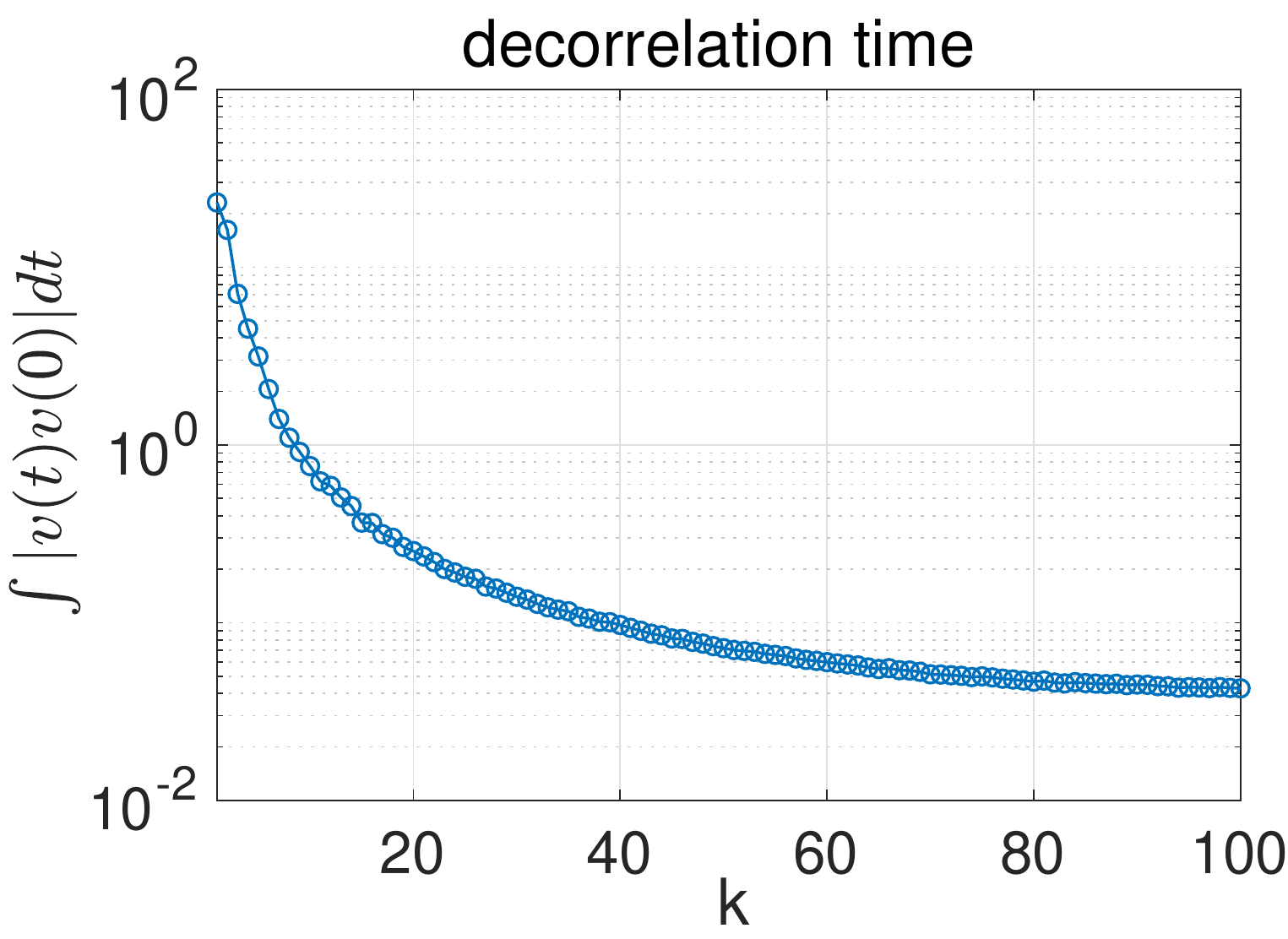}}\subfloat[time-series of energy in large and small scales]{\includegraphics[scale=0.3]{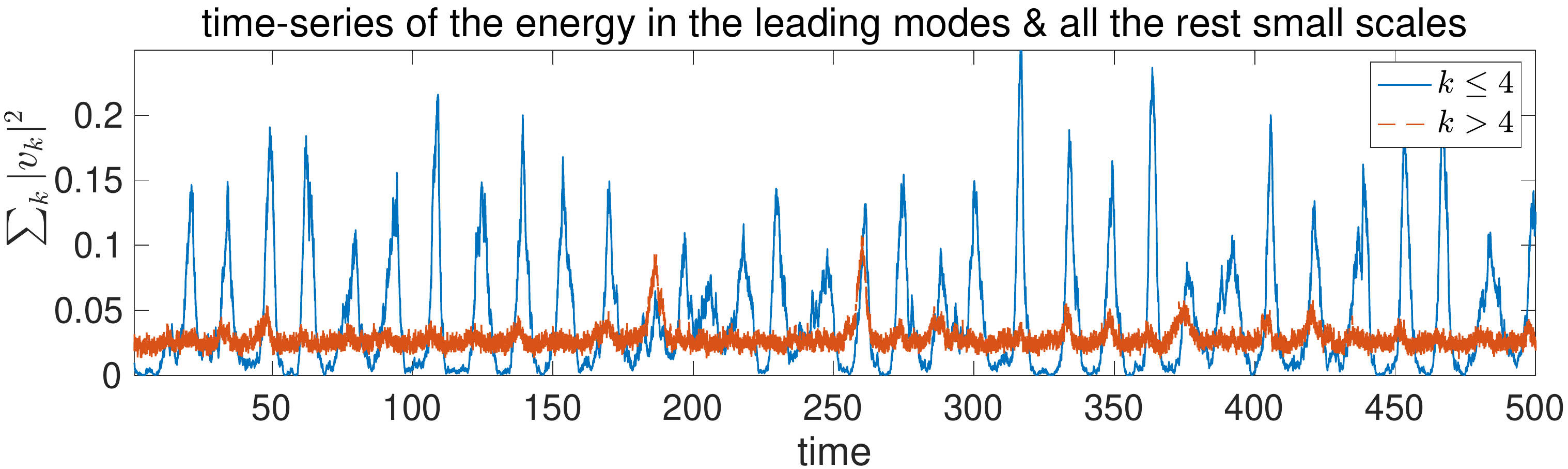}

}

\caption{Statistics of the conceptual turbulence model (\ref{eq:conceptual_model}).
Left: equilibrium energy spectrum and decorrelation time in the fluctuation
modes; Right: time series of the energy in the first 4 leading modes
and energy in all the rest fluctuation modes.\label{fig:Equilibrium-energy-spectrum}}
\end{figure}
First, we check the recovery of the solution trajectories using the
RBM model. Figure \ref{fig:Time-trajectory} plots the typical time-series
of the mean state $\bar{u}$ and the first 4 leading modes $v_{k}$
from direct simulation and the random batch approximation. Highly
non-Gaussian features with intermittent bursts of extreme events can
already be observed in the time-series of the leading modes. The intermittent
flow structure is a key feature to model that is closely related to
the largest-scale mean state. Qualitatively, we observe that the structures
in the most energetic modes are precisely captured from the RBM approach
especially with the random intermittent bursts and extreme events.
Notice that in each time step update of the model, only $p=5$ modes
are included in updating the mean equation randomly picked from the
total $K=100$ small-scale modes. It is shown that the batch size
can be further reduced to even $p=2$ in Figure \ref{fig:Statistical-errors}.
Thus the exhausting computational cost to resolve the long spectrum
of all fluctuation modes in each ensemble member at each time step
can be effectively avoided.

\begin{figure}
\includegraphics[scale=0.36]{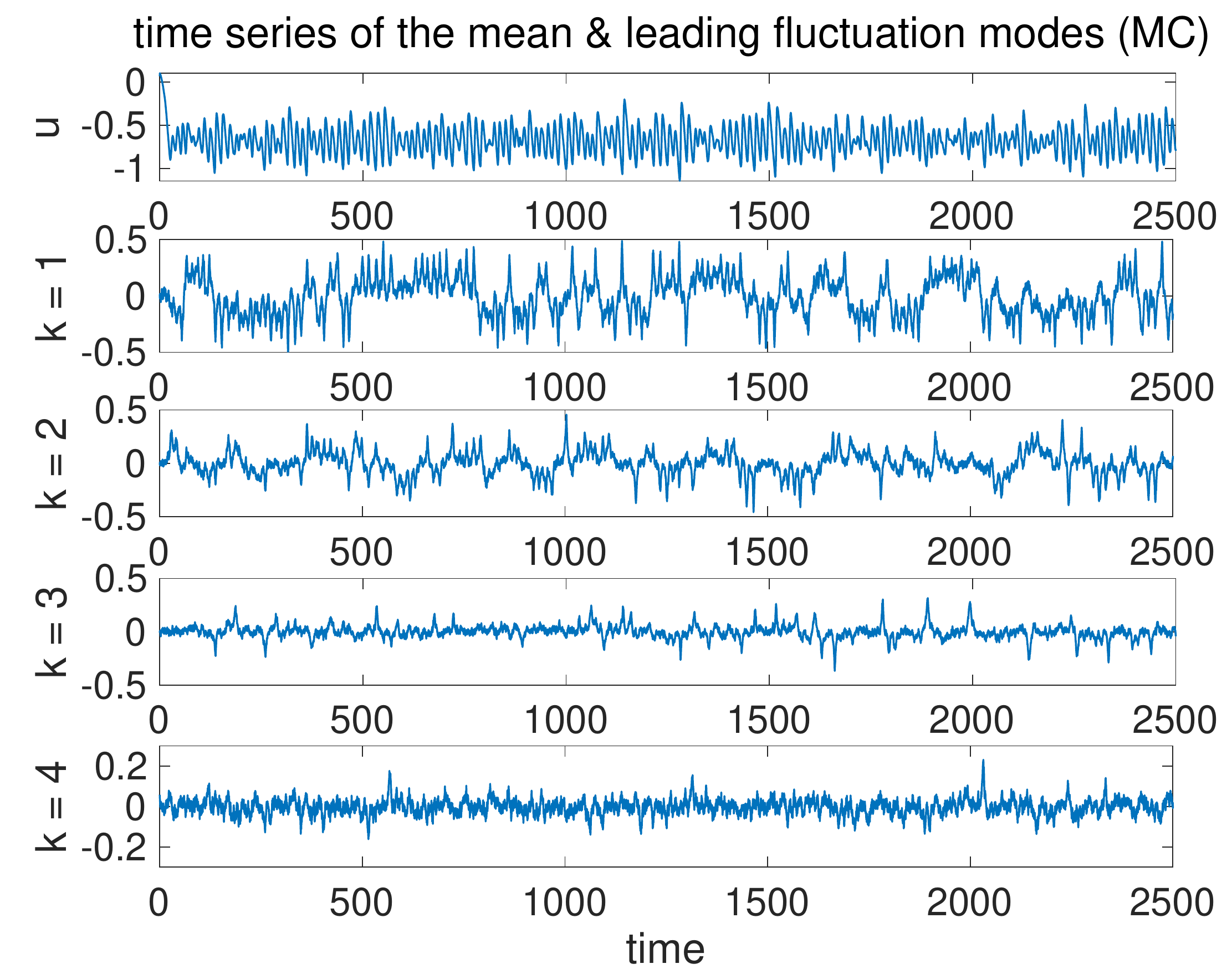}\includegraphics[scale=0.36]{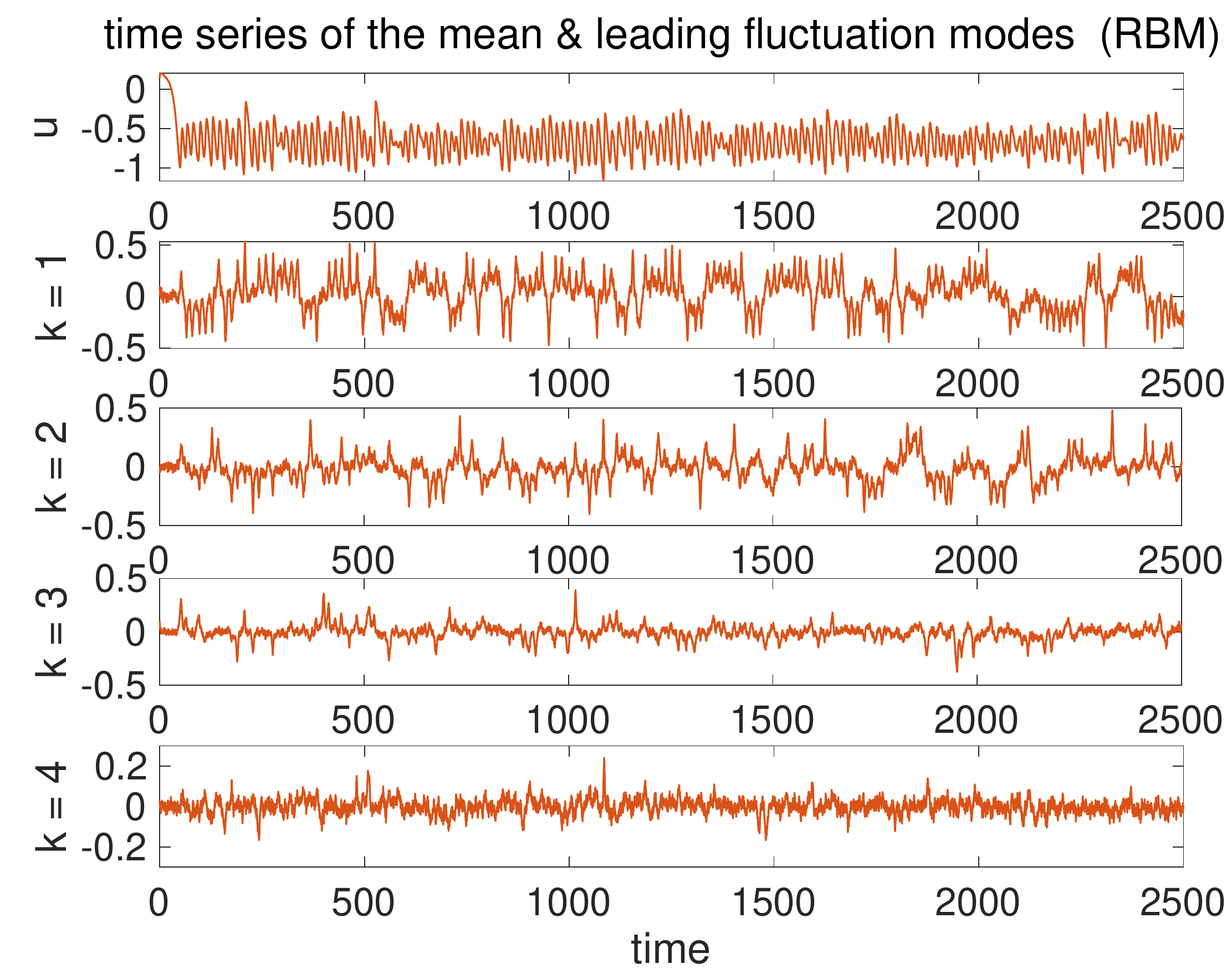}

\caption{Time trajectories of the mean state $\bar{u}$ and the first 4 leading
mode $v_{k}$ from the direct simulation (left) and the random batch
method (right).\label{fig:Time-trajectory}}

\end{figure}
Next, we test the performance of the RBM model in efficient ensemble
prediction of the probability distributions. In particular, the prediction
skill of both the final equilibrium probability distribution and the
transient probability distributions before equilibrium are considered.
In the tests, the initial samples at the starting time are drawn from
independent Gaussian distributions for each mode with mean zero and
a small variance. Thus the PDFs will go through a statistical transition
from the Gaussian initial distribution to the non-Gaussian final equilibrium.
Figure \ref{fig:Predicted-PDFs} compares the prediction of marginal
and joint PDFs in both the starting transient state and the final
statistical equilibrium. First from the sample distributions, we observe
the drastic deviation from Gaussian in the PDFs. Strongly non-Gaussian
structures appear in both the starting transient state and the final
statistical equilibrium. The fat-tailed PDFs refer to the intermittent
extreme flow structures and are of particular interest in the study
of turbulent flows. To capture such extreme features, a very large
sample size is required in the direct MC simulation to sufficiently
characterize the outliers that constitute the PDF tails. In contrast,
the RBM model focuses on the mean and the first 4 dominant modes and
recovers the joint PDFs of these most important modes. It is able
to capture both the transient and equilibrium non-Gaussian PDFs especially
the non-Gaussian PDF tails to a large extend accurately, while saving
the computational cost to a large degree requiring a much smaller
number of samples.

\begin{figure}
\includegraphics[scale=0.45]{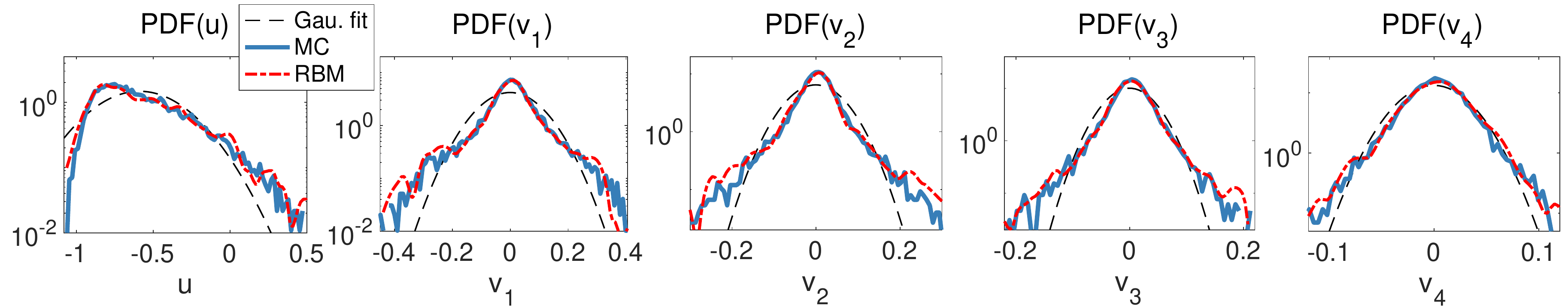}

\vspace{-1.em}

\subfloat[marginal and joint PDFs in starting transient state]{\includegraphics[scale=0.45]{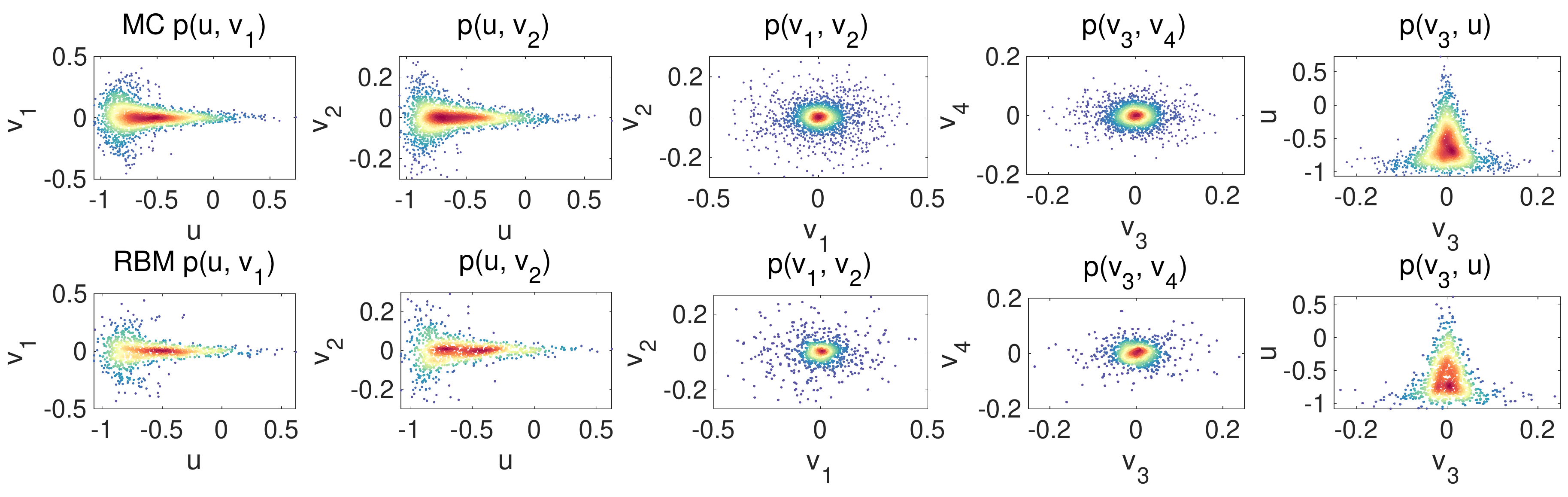}

}

\vspace{1.em}

\includegraphics[scale=0.45]{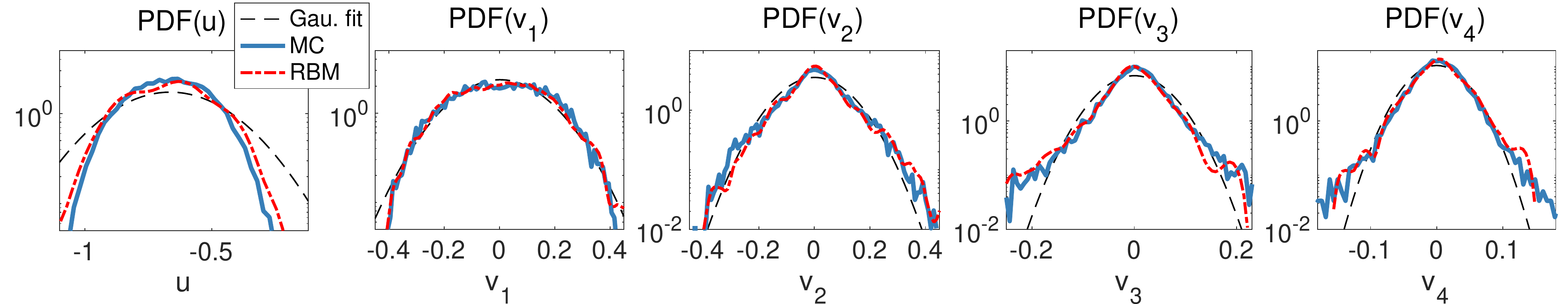}

\vspace{-1.em}

\subfloat[marginal and joint PDFs in final equilibrium state]{\includegraphics[scale=0.45]{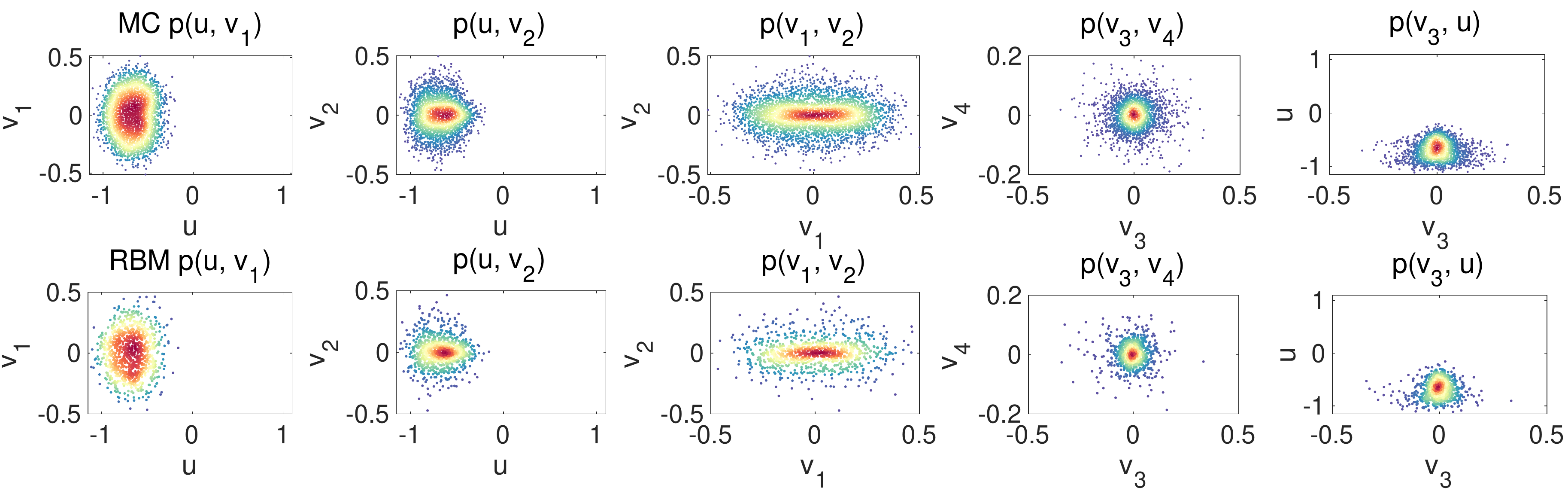}

}

\caption{Predicted PDFs from the direct MC simulation with $N=10000$ samples
and from the RBM model using $N_{1}=100$ samples for the conceptual
model. The 1D and 2D marginal PDFs of the leading modes $\left(\bar{u},v_{1},v_{2},v_{3},v_{4}\right)$
are compared. In the 1D PDFs, the Gaussian fit with the same mean
and variance are plotted in dashed line. The 2D joint PDFs are shown
by scatter plots with colors indicating sample density. The PDFs in
the starting transient state (upper) and the final equilibrium state
(lower) are compared.\label{fig:Predicted-PDFs}}

\end{figure}

\subsection{Random batch algorithm for the topographic barotropic model}

The topographic barotropic flow is a paradigm model in geophysical
turbulence \cite{qi2017low} which generates many representative features
found in real atmosphere and ocean. It identifies the multiscale interactions
between a large-scale mean flow and fluctuations through interaction
with topography. Under projection along one characteristic wavenumber
direction, the \emph{topographic barotropic model of layered topography}
can be expressed in term of a large-scale mean flow $U$ and a wide
spectrum of fluctuation modes $v_{k}$ in complex values as
\begin{equation}
\begin{aligned}\frac{\mathrm{d}U}{\mathrm{d}t}= & \:\frac{1}{K}\sum_{k=-K}^{K}h_{k}^{*}v_{k}-d_{0}U+\sigma_{0}\dot{W}_{0},\\
\frac{\mathrm{d}v_{k}}{\mathrm{d}t}= & \:i\left(k^{-1}\beta-kU\right)v_{k}-Uh_{k}-d_{k}v_{k}+\sigma_{k}\dot{W}_{k},\quad\left|k\right|\leq K
\end{aligned}
\label{eq:topo_model}
\end{equation}
Above, the complex flow modes satisfying $v_{-k}=v_{k}^{*}$ are derived
from the Fourier expansion of the fluctuation flow field. The model
(\ref{eq:topo_model}) constitutes a system of dimension $1+2K$.
The fixed complex modes $h_{k}$ represent the topographic effect,
and $\beta,d_{0},d_{k},\sigma_{0},\sigma_{k}$ are other model parameters
representing rotation, damping and unresolved forcing in large and
small scales. The topographic model eliminates the nonlinear interactions
between the fluctuation modes and focuses on the coupling effect between
the large and small scales. The equations (\ref{eq:topo_model}) becomes
consistent with the general multiscale model framework (\ref{eq:model_decoupled})
by introducing the auxiliary dynamics $\frac{\mathrm{d}h_{k}}{\mathrm{d}t}\equiv0$
which represents the constant topographic structure. Therefore, Algorithm
\ref{alg:Sample-RBM} still applies to implementing the RBM model
for this system. The topographic model also generates many key features
of interests in turbulent flows, including skewed non-Gaussian PDFs
and the related new type of extreme events. Different from the previous
conceptual model (\ref{eq:conceptual_model}), the small-scale feedback
to the mean $U$ is coupled through combined interaction with the
topographic mode $h_{k}$, and instability in the leading fluctuation
modes is introduced from the coupling with topographic stress. A detailed
derivation from the two-dimensional barotropic model and many desirable
properties such as conservation of energy can be found in \cite{majda2006nonlinear,majda1999simplified}.

The model parameters in (\ref{eq:topo_model}) are listed in Table
\ref{tab:Model-parameters-topographic}. The topographic modes $h_{k}$
are defined by the Fourier modes of the spatial topography structure
as a combination of two major large scales and multiple small-scale
random perturbations
\[
h=H\left(\sin x+\cos x\right)+\frac{H}{2}\left(\sin2x+\cos2x\right)+\sum_{3\leq\left|k\right|\leq K}e^{i\left(kx+\theta_{k}\right)},
\]
with $\theta_{k}$ the random phase shift drawn independently from
a uniform distribution in $\left[0,2\pi\right)$. A white noise forcing
with small amplitude $\sigma_{0}=\frac{1}{4\sqrt{2}}$ is also added
in large scale to induce stronger chaotic dynamics in the mean state.
This can also fit into the model framework with simple generalization.
In the small scales, $\beta=2$ represents the rotational effect of
the flow. Again we adopt the Kolmogorov spectrum $E_{k}=\frac{\sigma_{k}^{2}}{2d_{k}}=E_{0}k^{-5/3}$
in the first two complex modes $\left|k\right|\leq K_{1}=2$, and
with equipartition of energy for all the small scales, $\frac{\sigma_{k}^{2}}{2d_{k}}=E_{0}K_{1}^{-5/3},\left|k\right|>K_{1}$.
These choices of parameter values are following the non-dimensionalization
of the real physics measurements of the characteristic scales \cite{majda2006nonlinear}.
We first demonstrate the model statistical features in Figure \ref{fig:Equilibrium-energy-topo}.
The equilibrium energy spectrum and the decorrelation time in the
fluctuation modes display again the decaying energy and fast mixing
rate in small scales typical in turbulent flows. The small-scale fluctuation
modes contain small energy and fast decaying autocorrelation in each
single mode while they constitute a major contribution in their combined
feedback to the mean state. The bursts in the first two leading modes
imply the occurrence of extreme events excited by such multiscale
interaction. Thus the small scales play an non-negligible role and
are suitable for the RBM approximation.

\begin{table}
\begin{centering}
\begin{tabular}{cccccccccccccc}
\toprule 
$K$ & $H$ & $\beta$ &  & $d_{0}$ & $\sigma_{0}$ & $d_{k}$ & $\sigma_{k}$ & $E_{0}$ &  & $K_{1}$ & $p$ & $N$ (full MC) & $N_{1}$ (RBM)\tabularnewline
\midrule 
100 & 1 & 2 &  & 0.0125 & $\frac{1}{4\sqrt{2}}$ & $0.0125k$ & $\sqrt{2E_{k}d_{k}}$ & 0.02 &  & 2 & 5 & 10000 & 100\tabularnewline
\bottomrule
\end{tabular}
\par\end{centering}
\caption{Parameter values for the conceptual turbulent model.\label{tab:Model-parameters-topographic}}
\end{table}
\begin{figure}
\subfloat[equilibrium spectrum and decorrelation time]{\includegraphics[scale=0.25]{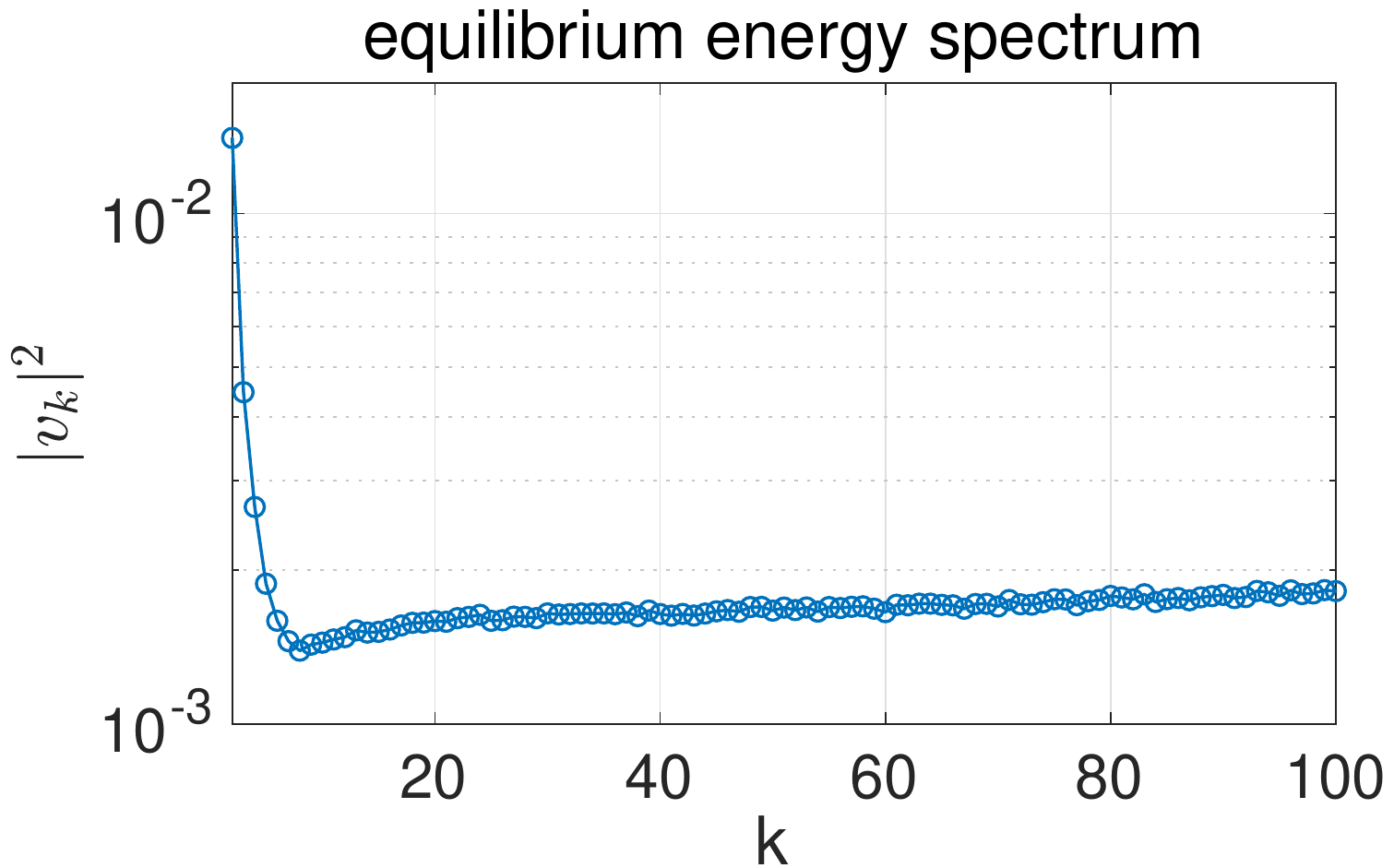}\includegraphics[scale=0.25]{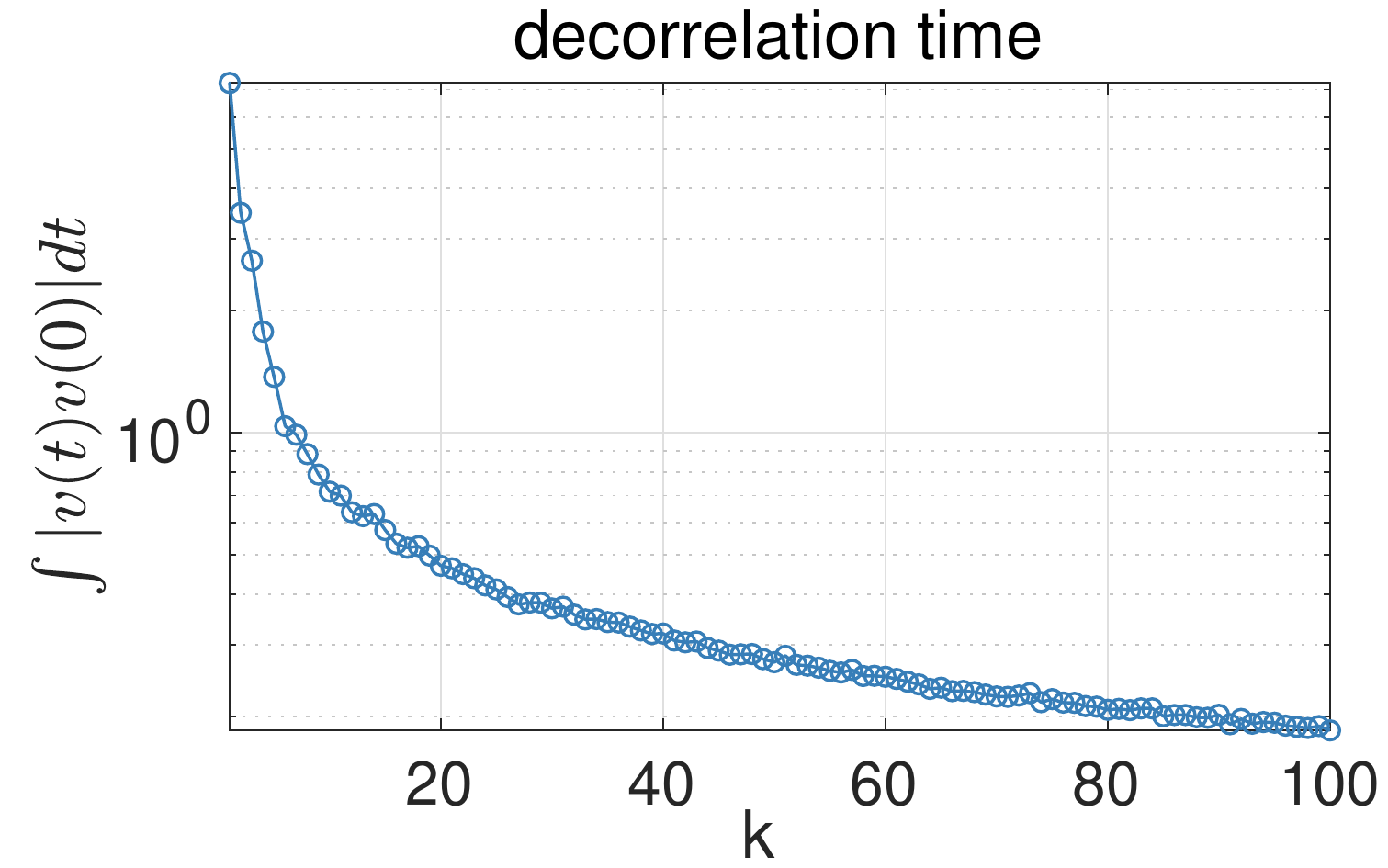}

}\hspace{-.5em}\subfloat[time-series of energy in large and small scales]{\includegraphics[scale=0.33]{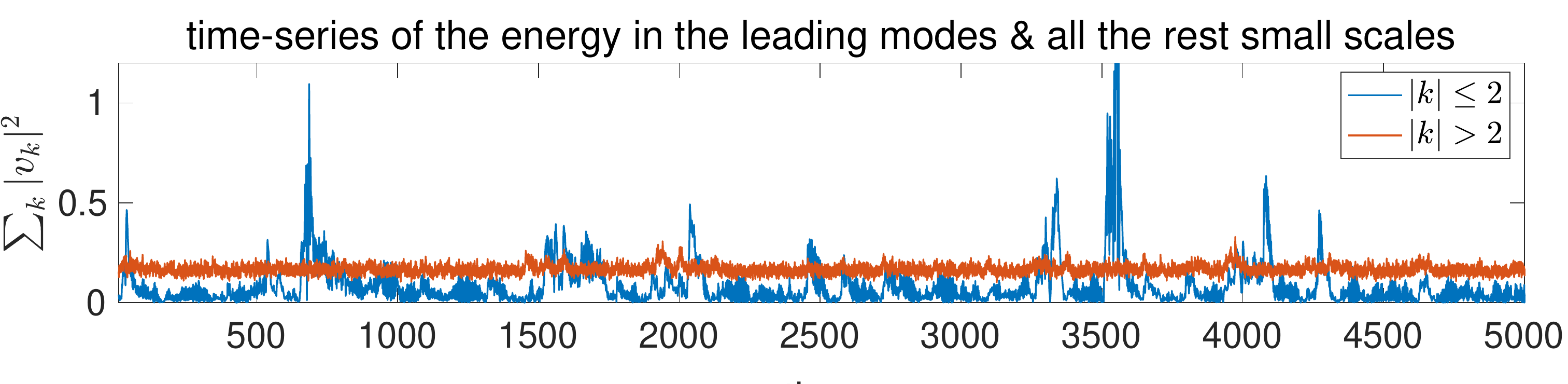}

}

\caption{Statistics of the topographic barotropic model (\ref{eq:topo_model}).
Left: equilibrium energy spectrum and decorrelation time in the fluctuation
modes; Right: time-series of the energy in the first 2 leading modes
and energy in all the rest fluctuation modes.\label{fig:Equilibrium-energy-topo}}
\end{figure}
In the test for ensemble forecast, we aim to capture the PDFs in the
mean flow state $U$ together with the first two leading modes $v_{1},v_{2}$.
Since the fluctuation modes are in complex values, it forms a 5-dimensional
subspace compared with the full model dimension $1+2K=201$. A large
ensemble size $N=1\times10^{4}$ is again needed to sufficiently sample
the high-dimensional phase space of the system in the full MC model,
while in contrast the RBM model uses a much smaller ensemble of $N_{1}=100$
samples. We put the detailed RBM model formulation of the topographic
model (\ref{eq:topo_model}) following Algorithm \ref{alg:Sample-RBM}
in Appendix \ref{subsec:RBM-topographic}. In this topographic model,
one typical feature is the intermittent bursts of extreme events in
both the mean flow $U$ and the leading fluctuation modes $v_{1},v_{2}$,
reflected by the skewed fat-tails in the resulting PDFs. This makes
a even more challenging case for accurate ensemble forecast since
it usually requires a much larger ensemble size to capture the extreme
events in the asymmetric PDF tails with accuracy. The RBM prediction
for the marginal PDFs and the joint PDFs in the resolved leading modes
is shown in Figure \ref{fig:Predicted-PDFs-topo} compared with the
direct MC results. As shown in both the marginal PDFs of the leading
modes and the joint distributions, complicated non-Gaussian structures
are generated during the evolution of the states. Again, the RBM model
maintains the high skill to capture the skewed PDF structures while
greatly reducing the computational cost.

\begin{figure}
\subfloat{\includegraphics[scale=0.45]{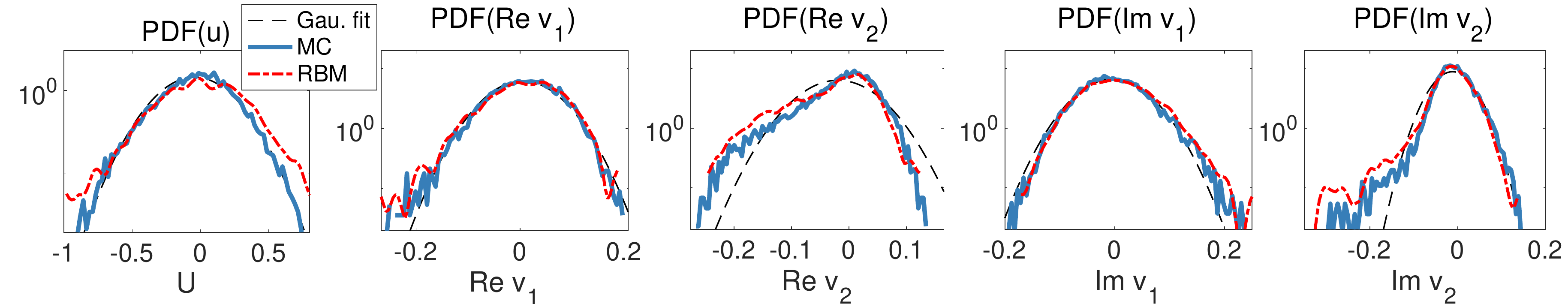}}

\vspace{-1.em}

\subfloat{\includegraphics[scale=0.45]{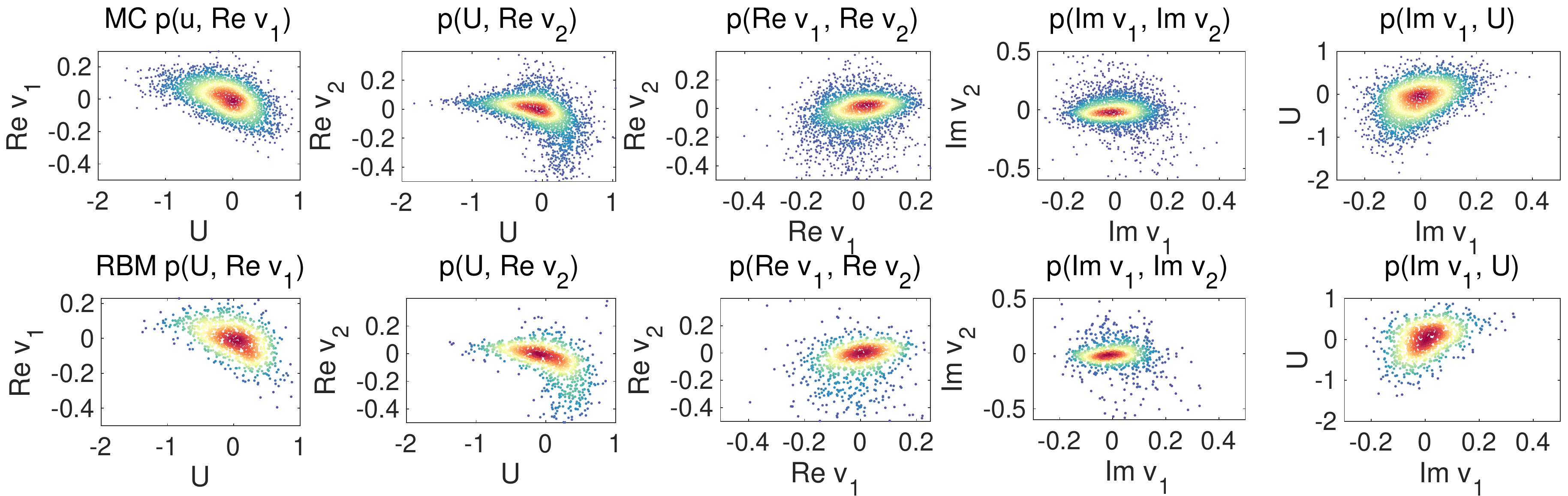}}

\caption{Predicted PDFs from the direct MC simulation with $N=10000$ samples
and from the RBM model using $N_{1}=100$ samples for the topographic
model. The 1D and 2D marginal PDFs of the leading modes $\left(U,\mathfrak{Re}v_{1},\mathfrak{Re}v_{2},\mathfrak{Im}v_{1},\mathfrak{Im}v_{2}\right)$
are compared. In the 1D PDFs, the Gaussian fit with the same mean
and variance are plotted in dashed line. The 2D joint PDFs are shown
by scatter plots with colors indicating sample density.\label{fig:Predicted-PDFs-topo}}
\end{figure}
Finally, we offer a quantitative quantification for the prediction
errors by measuring the empirical statistics in the first two moments
as $E_{2}=\left|\frac{1}{N}\sum_{i}\left|u^{\left(i\right)}\right|^{2}-\frac{1}{N}\sum_{i}\left|\tilde{u}^{\left(i\right)}\right|^{2}\right|$
with $u$ the mean state of the full MC model and $\tilde{u}$ the
RBM solution. A large sample size $N=1\times10^{4}$ is used to reduce
the error from the empirical sample average approximation. In Figure
\ref{fig:Statistical-errors}, we plot the evolution of errors for
the two test models. Notice that there still exist errors from the
ensemble approximation of the expectation due to the finite sample
size $N$. It is observed that the RBM model maintains accurate prediction
skill of the statistics with small errors during the model evolution.
The errors gradually grow in time and will saturate at a low level
when the system reaches statistical equilibrium. As a further comparison,
we also compare the errors under different batch sizes $p$. It shows
that we can even push the batch sizes to an extreme $p=2$ and the
model still provides accurate prediction with just a slightly larger
error. Overall, this confirms the robust performance of the RBM model
subject to different turbulent dynamical features and for different
statistical regime in both transient and final equilibrium state.

\begin{figure}
\subfloat[conceptual model]{\includegraphics[scale=0.38]{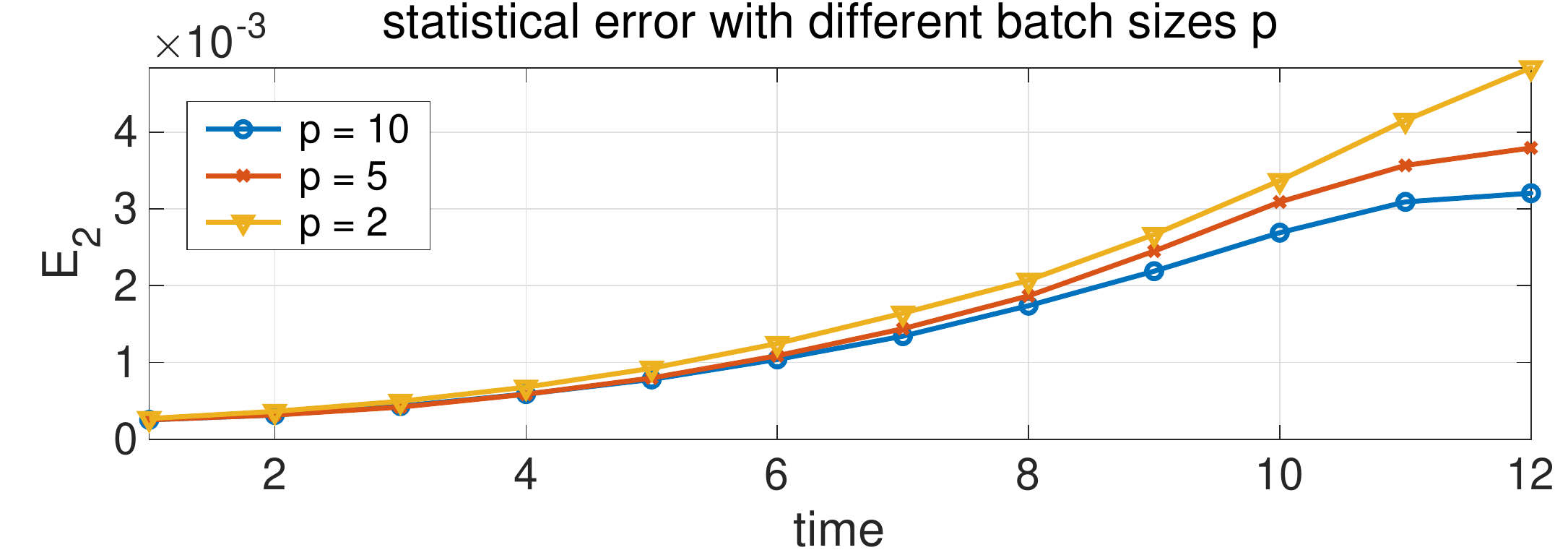}

}\subfloat[topographic model]{\includegraphics[scale=0.38]{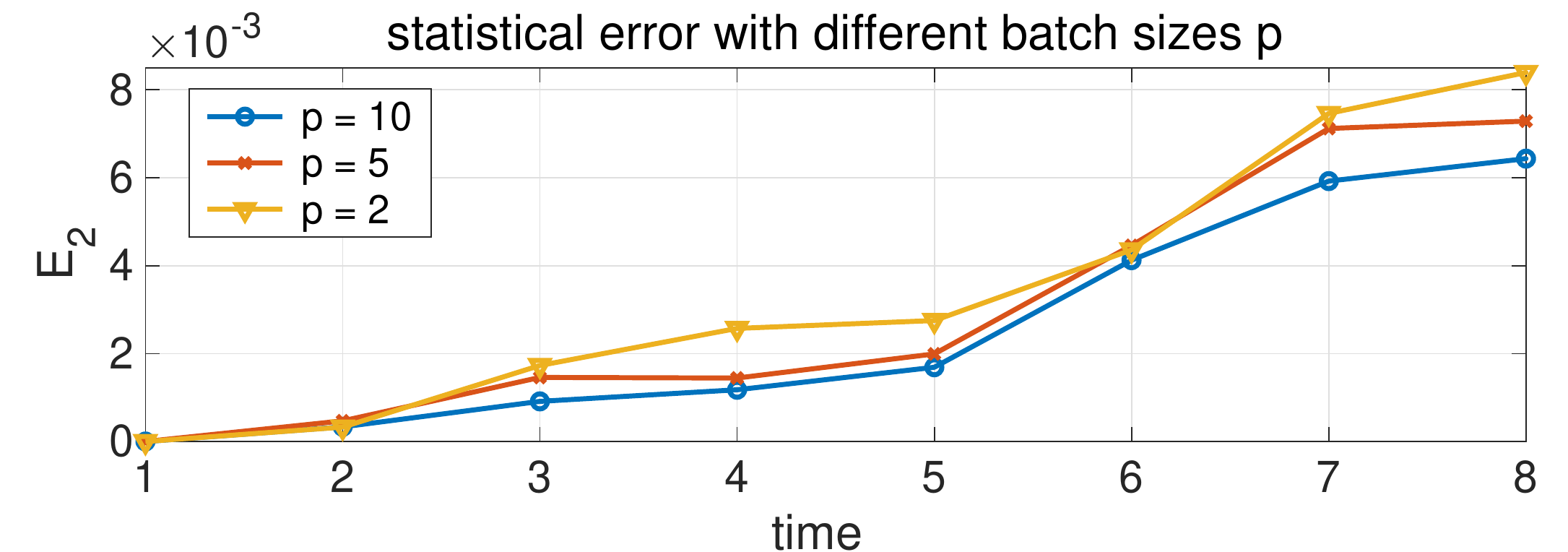}

}

\caption{Statistical errors $E_{2}=\left|\frac{1}{N}\sum_{i}\left|u^{\left(i\right)}\right|^{2}-\frac{1}{N}\sum_{i}\left|\tilde{u}^{\left(i\right)}\right|^{2}\right|$
of the RBM model in the time development of the conceptual model (left)
and the topographic model (right). The results with different batch
sizes $p=2,5,10$ are also compared.\label{fig:Statistical-errors}}

\end{figure}

\section{Summary\label{sec:Summary}}

We developed a new efficient ensemble prediction strategy to model
and forecast the time evolution of probability distributions and the
associated statistical features in the dominant large-scale states
of complex turbulent systems with coupled multiscale structure. Standard
Monte-Carlo simulation of a high-dimensional turbulent system suffers
the curse-of-dimensionality thus requires an unaffordable ensemble
size to even maintain a low-order approximation. The proposed RBM
model circumvents the inherent difficulty by just sampling a low-dimensional
subspace which contains the dominant large-scale states, while the
contributions from all the small-scale fluctuation modes are fully
considered through a random batch decomposition. The wide spectrum
of the $K$ small-scale fluctuation modes is randomly divided into
small batches of size $p$ at the start of each time updating step,
and each of the $K/p$ batches is associated with one of the $N$
large scale-state samples to update the coupled nonlinear feedback
term computed inside the batch. The modes in each batch serves as
the different samples to compute the small-scale feedback in the ensemble
update of the large-scale state. The computational cost is then greatly
reduced based on the random batch decomposition which avoids the expensive
ensemble simulation of the large number of small-scale fluctuation
modes. The true multiscale dynamics is recovered in the efficient
algorithm due to the frequent resampling of the batches at each time
updating step and the fast mixing rate of the ergodic small-scale
fluctuation modes. The resulting algorithm is also very easy to implement
for a general group of multiscale turbulent models capable of creating
a wide variety of realistic complex phenomena. Therefore, the efficient
RBM model developed here provides a useful tool for improving the
understanding of various turbulent features observed in natural and
engineering systems, and the further development of effective methods
in uncertainty quantification and data assimilation \cite{reich2015probabilistic,majda2016introduction}
of complex turbulent systems.

In the analysis of the RBM model for the coupled large and small scale
turbulent systems, the approximation errors under statistical expectation
of the empirical ensemble average are derived by comparing the semigroups
generated by the backward equation of the original model and the RBM
approximation. The error is shown to be only related to the numerical
time step and independent of the sample size and the full dimension
of the system. The RBM model is then applied to two representative
turbulent models with close link to several realistic phenomena such
as extreme events and intermittent instability. One central issue
in practical forecast of turbulent systems concerns the accurate characterization
of extreme events represented in the long extended PDF tails and the
deviation from the Gaussian distribution. The RBM model is show to
have uniformly high skill in predicting various different structures
in the PDFs during the time evolution in both test models driven by
different types of coupling mechanisms. Only a very small ensemble
size $N=100$ is needed to achieve sufficient accuracy for models
with a high dimensional fluctuation state dimension $K=100$ and $K=200$.
In contrast, the direct MC simulation requires at least $N=1\times10^{4}$
samples to reach a relatively high accuracy. In the future development
in more realistic applications, the RBM model shows potential to overcome
the curse-of-dimensionality for a wider group of practical problems
involving fully turbulent high dimensional flows.

\section*{Acknowledge}

The research of J.-G. L. is partially supported under the NSF grant
No. DMS-2106988. The research of D.Q. is partially supported by the
start-up funds and the PCCRC Seed Funding provided by Purdue University.

\appendix
\renewcommand\theequation{A\arabic{equation}}
\setcounter{equation}{0}

\section{Proofs of the Lemmas\label{appen1:Proofs-of-lemmas}}
\begin{proof}
[Proof of Lemma \ref{lem:coeff_exp}]This is the direct conclusion
by counting the number of ordered combinations of the $n$ random
batches. First, we define the total number of ways of listing $np$
distinguishable objects into $n$ ordered batches of size $p$ as
\[
n!M\left(n\right)=\frac{\left(np\right)!}{\left(p!\right)^{n}}.
\]
We use $M\left(n\right)$ to denote the number of combinations to
put $np$ objects to $p$ groups without order. Next, for the $k$-th
mode falling in the batch $i$, we determine the $i$-th batch to
contain $k$ and select the other $p-1$ objects in this batch from
the remaining $np-1$ objects, then order the rest $n-1$ batches.
This gives
\[
\mathbb{E}I_{i}\left(k\right)=\frac{\begin{pmatrix}np-1\\
p-1
\end{pmatrix}\left(n-1\right)!M\left(n-1\right)}{n!M\left(n\right)}=\frac{1}{n}=\frac{p}{K}.
\]
Similarly, for two modes $k\ne l$ falling in the same batch $i$,
we put these two objects together with other $p-2$ objects, then
still order the rest $n-1$ batches, which gives
\[
\mathbb{E}I_{i}\left(k\right)I_{i}\left(l\right)=\frac{\begin{pmatrix}np-2\\
p-2
\end{pmatrix}\left(n-1\right)!M\left(n-1\right)}{n!M\left(n\right)}=\frac{1}{n}\frac{p-1}{np-1}=\frac{p}{K}\frac{p-1}{K-1}.
\]
\end{proof}
\
\begin{proof}
[Proof of Lemma \ref{lem:flucs_bound}]Applying Ito's Lemma for $f\left(z\right)=\left\Vert z\right\Vert ^{2q}=\left(\sum_{k}\left|z_{k}\right|^{2}\right)^{q}$
according to the SDE of $Z_{k}$ in (\ref{eq:mc_analysis}), we have
the conditional expectation $\mathbb{E}_{u}$ with a fixed $u$ as
\begin{align*}
\frac{\mathrm{d}}{\mathrm{d}t}\mathbb{E}_{u}\left\Vert Z\right\Vert ^{2q}= & 2q\sum_{k}\left(\gamma_{k}\left(u\right)-d_{k}\right)\mathbb{E}_{u}Z_{k}^{2}\left\Vert Z\right\Vert ^{2\left(q-1\right)}\\
 & +q\mathbb{E}_{u}\left\Vert Z\right\Vert ^{2\left(q-1\right)}\sum_{k}\sigma_{k}^{2}+2q\left(q-1\right)\sum_{k}\sigma_{k}^{2}\mathbb{E}_{u}Z_{k}^{2}\left\Vert Z\right\Vert ^{2\left(q-2\right)}.
\end{align*}
From (\ref{eq:assump3}) in Assumption \ref{assu:assump_coeffs},
the coefficients are uniformly bounded, $d_{k}-\gamma_{k}\left(u\right)\geq r>0$
and $\sum_{k}\sigma_{k}^{2}\leq C$. Thus
\[
\frac{\mathrm{d}}{\mathrm{d}t}\mathbb{E}_{u}\left\Vert Z\right\Vert ^{2q}\leq-2qr\mathbb{E}_{u}\left\Vert Z\right\Vert ^{2q}+q^{2}C^{\prime}\mathbb{E}_{u}\left\Vert Z\right\Vert ^{2\left(q-1\right)}.
\]
First for $q=1$, the last term on the above inequality becomes a
constant, thus $\mathbb{E}\left\Vert Z_{t}\right\Vert ^{2}=\mathbb{E}^{U}\mathbb{E}\left\Vert Z\left(t\mid U\right)\right\Vert ^{2}\leq\frac{C^{\prime}}{2r}\equiv C_{1}$.
Next, we have by induction $\mathbb{E}\left\Vert Z_{t}\right\Vert ^{2q}\leq C_{q}$
for any integer $q\geq1$. In the same way, we have $\mathbb{E}\left\Vert \tilde{Z}_{t}\right\Vert ^{2q}<C_{q}$
for any $t>0$.
\end{proof}
\
\begin{proof}
[Proof of Lemma \ref{lem:action_bound}]By using the backward Kolmogorov
equation \eqref{eq:backward} for $w_{z}\left(x,t\right)$ and taking
its derivative about the $i$-th coordinate $x_{i}\in\mathbb{R}^{d}$
with $i=1,\cdots,N$, it yields
\[
\partial_{t}\partial_{x_{i}}w_{z}=\mathcal{L}_{z}\partial_{x_{i}}w_{z}+\nabla V\left(x_{i}\right)\cdot\partial_{x_{i}}w_{z}.
\]
Above, we define the vector function $\partial_{x_{i}}w_{z}=\left\{ \partial_{x_{i}^{j}}w_{z}\right\} _{j=1}^{d}\in\mathbb{R}^{d}$
with $x_{i}=\left\{ x_{i}^{j}\right\} _{j}$ and $\nabla V\left(x\right)=\left\{ \partial_{x^{j}}V_{j}\right\} _{j}:\mathbb{R}^{d}\rightarrow\mathbb{R}^{d\times d}$.
From the definition of the generator, the only term that is dependent
on $x$ in $\mathcal{L}_{z}$ is $V\left(x_{i}\right)$. This leads
to the formal solution 
\[
\partial_{x_{i}}w_{z}\left(x,t\right)=e^{t\mathcal{L}_{z}}\partial_{x_{i}}w_{z}\left(x,0\right)+\int_{0}^{t}e^{\left(t-s\right)\mathcal{L}_{z}}\nabla V\left(x_{i}\right)\cdot\partial_{x_{i}}w_{z}\left(x,s\right)\mathrm{d}s.
\]
By the contraction of the semigroup $e^{t\mathcal{L}_{z}}$, we have
\[
\left\Vert e^{t\mathcal{L}_{z}}w_{z}\left(x,0\right)\right\Vert _{\infty}\leq\left\Vert w_{z}\left(x,0\right)\right\Vert _{\infty},
\]
and 
\[
\left\Vert e^{\left(t-s\right)\mathcal{L}_{z}}\nabla V\left(x_{i}\right)\cdot\partial_{x_{i}}w_{z}\left(x,s\right)\right\Vert _{\infty}\leq\left\Vert \nabla V\left(x_{i}\right)\cdot\partial_{x_{i}}w_{z}\left(x,s\right)\right\Vert _{\infty}.
\]
Therefore, using the uniform boundedness of $V$ in (\ref{eq:assump2})
\[
\begin{aligned}\left\Vert \partial_{x_{i}}w_{z}\left(x,t\right)\right\Vert _{\infty} & \leq\left\Vert \partial_{x_{i}}w_{z}\left(x,0\right)\right\Vert _{\infty}+\int_{0}^{t}\left\Vert \nabla V\left(x_{i}\right)\cdot\partial_{x_{i}}w_{z}\left(x,s\right)\right\Vert _{\infty}\mathrm{d}s\\
 & \leq\left\Vert \partial_{x_{i}}w\left(x,0\right)\right\Vert _{\infty}+C\int_{0}^{t}\left\Vert \partial_{x_{i}}w_{z}\left(x,s\right)\right\Vert _{\infty}\mathrm{d}s.
\end{aligned}
\]
Using the integral form of Gr\"{o}nwall's inequality, we get
\[
\left\Vert \partial_{x_{i}}w\left(x,t\right)\right\Vert _{\infty}\leq\mathbb{E}_{Z}\left\Vert \partial_{x_{i}}w_{Z}\left(x,t\right)\right\Vert _{\infty}\leq C\left(t\right)\left\Vert \partial_{x_{i}}w\left(x,0\right)\right\Vert _{\infty}\leq\frac{C\left(t,\varphi\right)}{N}.
\]
Above in the last inequality, by definition $w\left(x,0\right)=\frac{1}{N}\sum_{i=1}^{N}\varphi\left(x_{i}\right)$,
then $\partial_{x_{i}}w\left(x,0\right)=\frac{1}{N}\nabla\varphi\left(x_{i}\right)$
is bounded since $\varphi\in C_{b}^{2}$.

Next, by taking a second derivative about $x_{j}$ on the backward
equation, we have
\[
\partial_{t}\partial_{x_{i}x_{j}}^{2}w_{z}=\mathcal{L}_{z}\partial_{x_{i}x_{j}}^{2}w_{z}+\left[\nabla V\left(x_{i}\right)+\nabla V\left(x_{j}\right)\right]\cdot\partial_{x_{i}x_{j}}^{2}w_{z}+\delta_{ij}\nabla^{2}V\left(x_{i}\right)\cdot\partial_{x_{i}}w_{z}.
\]
Above, the last term on the right hand side only appears when $i=j$.
Following the same argument as before, and using the boundedness of
the initial condition $\partial_{x_{i}x_{j}}^{2}w\left(x,0\right)=\frac{1}{N}\nabla^{2}\varphi\left(x_{i}\right)\delta_{ij}$
and $\nabla^{2}V$, we first get the bound for second derivative about
$i,j$
\[
\left\Vert \partial_{x_{i}x_{j}}^{2}w\left(x,t\right)\right\Vert _{\infty}\leq\frac{C}{N}.
\]
Then under similar estimation of the second order derivation equation
and taking summation among all the samples $i,j=1,\cdots,N$
\begin{align*}
\sum_{i,j}\left\Vert \mathbb{E}^{Z}\partial_{x_{i}x_{j}}^{2}w_{Z}\left(x,t\right)\right\Vert _{\infty} & \leq\sum_{i,j}\left\Vert \partial_{x_{i}x_{j}}^{2}w\left(x,0\right)\right\Vert _{\infty}+C_{0}\int_{0}^{t}\sum_{i}\left\Vert \mathbb{E}^{Z}\partial_{x_{i}}w_{Z}\left(x,s\right)\right\Vert _{\infty}\mathrm{d}s+C_{1}\int_{0}^{t}\sum_{i,j}\left\Vert \mathbb{E}^{Z}\partial_{x_{i}x_{j}}^{2}w_{Z}\left(x,s\right)\right\Vert _{\infty}\mathrm{d}s\\
 & \leq C\left(T\right)+C_{1}\int_{0}^{t}\sum_{i,j}\left\Vert \mathbb{E}^{Z}\partial_{x_{i}x_{j}}^{2}w_{Z}\left(x,s\right)\right\Vert _{\infty}\mathrm{d}s.
\end{align*}
In the second row, we use the initial condition $w\left(x,0\right)=\frac{1}{N}\sum_{i=1}^{N}\varphi\left(x_{i}\right)$
and the estimation on the first derivative, so that the bounds on
the right are both of order $O\left(1\right)$ and independent of
$N$
\begin{align*}
\sum_{i,j}\left\Vert \partial_{x_{i}x_{j}}^{2}w\left(\cdot,0\right)\right\Vert _{\infty} & =\sum_{i}\frac{1}{N}\left\Vert \nabla^{2}\varphi\right\Vert _{\infty},\\
\sum_{i}\left\Vert \partial_{x_{i}}w\left(\cdot,s\right)\right\Vert _{\infty} & \leq\sum_{i}\frac{C}{N}.
\end{align*}
Letting $f\left(t\right)=\sum_{i,j}\left\Vert \mathbb{E}^{Z}\partial_{x_{i}x_{j}}^{2}w_{Z}\left(\cdot,t\right)\right\Vert _{\infty}$,
then
\[
f\left(t\right)\leq C+C_{1}\int_{0}^{t}f\left(s\right)\mathrm{d}s.
\]
Gr\"{o}nwall's inequality gives
\[
\sum_{i,j}\left\Vert \partial_{x_{i}x_{j}}^{2}w\right\Vert _{\infty}\leq C\left(t,\varphi\right),
\]
with the constant $C$ on the right only dependent on the time $t$
and the test function $\varphi$.
\end{proof}
\renewcommand\theequation{B\arabic{equation}}
\setcounter{equation}{0}

\section{Detailed RBM formulation for the test models\label{appen2:Detailed-RBM-formulation}}

\subsection{RBM equations for the conceptual turbulent model\label{subsec:RBM-conceptual}}

Here, we show the detailed equations of the RBM model for the ensemble
prediction of the conceptual turbulent model (\ref{eq:conceptual_model}).
First, we decompose the small-scale states $v=\left\{ v_{1,k},v_{2,k}\right\} $
further into the unstable leading modes $v_{1,k},k\leq K_{1}$, and
the rest fluctuating smaller scales, $v_{2,k},K_{1}<k\leq K$, with
stable dynamics. This decomposition is used to also resolve the leading
fluctuation modes containing intermittent unstable growth when $-\left(d_{k}+\gamma\bar{u}\right)>0$
from the coupling with the mean. In this way, we still only need to
sample a much lower dimensional subspace containing the most energetic
leading modes together with the mean state, that is, $\left\{ \bar{u},v_{1,k}\right\} _{k\leq K_{1}}$,
while the much less energetic stable smaller-scale fluctuation modes
are modeled by the random batches. The ensemble approximation of the
marginal PDF with $N_{1}$ samples becomes
\[
p_{\mathrm{RBM}}\left(\bar{u},v_{1}\right)=\frac{1}{N_{1}}\sum_{i=1}^{N_{1}}\delta\left(\bar{u}-\bar{u}^{\left(i\right)}\right)\otimes\prod_{k=1}^{K_{1}}\delta\left(v_{1,k}-v_{1,k}^{\left(i\right)}\right).
\]
The above empirical approximation requires a much smaller ensemble
size $N_{1}$ and is independent of the full dimension $K\left(\gg K_{1}\right)$
of the system. Next, the random batch partition is applied to the
large group of small-scale fluctuation modes $\left\{ v_{2,k}\right\} _{k>K}$
exploiting their ergodicity with fast mixing rate. In particular,
the large number of fluctuation modes $\left\{ v_{2,k}\right\} $
are divided into small batches of size $p$ each. In the RBM model,
the sample size $N_{1}$ in the ensemble simulation is associated
with the same number of batches. For the $i$-th sample in the large
scale ensemble $\left\{ \bar{u}^{\left(i\right)},v_{1,k}^{\left(i\right)}\right\} $,
the large-scale mean equation only includes $p$ randomly picked small-scale
modes in one batch, $\left\{ v_{2,k}:k\in\mathcal{I}_{i}\right\} $.
The RBM model with samples $i=1,\cdots,N_{1}$ at time step $t=t_{n}$
becomes
\begin{equation}
\begin{aligned}\frac{\mathrm{d}\bar{u}^{\left(i\right)}}{\mathrm{d}t}= & -\bar{d}\bar{u}^{\left(i\right)}+\frac{\gamma}{K_{1}}\sum_{k=1}^{K_{1}}\left(v_{1,k}^{\left(i\right)}\right)^{2}+\frac{\gamma}{p}\sum_{k\in\mathcal{I}_{i}}\left(v_{2,k}\right)^{2}-\bar{\alpha}\left(\bar{u}^{\left(i\right)}\right)^{3}+\bar{F},\\
\frac{\mathrm{d}v_{1,k}^{\left(i\right)}}{\mathrm{d}t}= & -d_{k}v_{1,k}^{\left(i\right)}-\gamma\bar{u}^{\left(i\right)}v_{1,k}^{\left(i\right)}+\sigma_{k}\dot{W}_{k}^{\left(i\right)},\quad1\leq k\leq K_{1},\\
\frac{\mathrm{d}v_{2,k}}{\mathrm{d}t}= & -d_{k}v_{2,k}-\gamma\bar{u}^{\left(i\right)}v_{2,k}+\sigma_{k}\dot{W}_{k},\quad\quad k\in\mathcal{I}_{i},\quad i=1,\cdots,N_{1}.
\end{aligned}
\label{eq:model_batch}
\end{equation}
The small-scale modes $v_{2,k}$ are segmented into small batches
with $\cup_{i}\mathcal{I}_{i}=\left\{ k:K_{1}<k\leq K\right\} $.
In addition, we add a small ensemble $n_{2}\left(=5\right)$ for the
small-scale modes. Therefore, there are $N_{1}=\left\lceil n_{2}\frac{K-K_{1}}{p}\right\rceil $
samples to sufficiently sample the resolved subspace. Only a very
small ensemble is needed for the high dimensional subspace for $\left\{ v_{2,k}\right\} $
satisfying

\[
n_{2}=\frac{N_{1}p}{K-K_{1}}.
\]
We can estimate the computational cost of (\ref{eq:model_batch})
as $O\left(N_{1}\left(1+K_{1}\right)p\right)\sim O\left(n_{2}\left(1+K_{1}\right)K\right)$.
Notice that the cost won't have the exponential growth depending on
the full dimension $1+K$ of the system by avoiding sampling the full
high dimensional space.

\subsection{RBM equations for the topographic barotropic model\label{subsec:RBM-topographic}}

In a similar way, we display the detailed equations for the implementation
of the RBM model for the topographic barotropic model (\ref{eq:topo_model}).
Again, we focus on the ensemble approximation of the dominant mean
state and first $K_{1}$ leading modes, $\left(U,v_{1,k}\right)$
with $\left|k\right|\leq K_{1}$
\[
p_{\mathrm{RBM}}\left(U,v_{1}\right)=\frac{1}{N_{1}}\sum_{i=1}^{N_{1}}\delta\left(U-U^{\left(i\right)}\right)\otimes\prod_{k=1}^{K_{1}}\delta\left(v_{1,k}-v_{1,k}^{\left(i\right)}\right),
\]
while all the other small-scale modes $\left\{ v_{2,k}\right\} ,K_{1}<\left|k\right|\leq K$
are modeled in the random batches $\cup_{i}\mathcal{I}_{i}=\left\{ k:K_{1}<\left\lceil k\right\rceil \leq K\right\} $
with size of the batches $p=\left|\mathcal{I}_{i}\right|$. The RBM
model for samples $i=1,\cdots,N_{1}$ at time step $t=t_{n}$ becomes

\begin{equation}
\begin{aligned}\frac{\mathrm{d}U^{\left(i\right)}}{\mathrm{d}t}= & \:\frac{1}{K_{1}}\sum_{\left|k\right|\leq K_{1}}h_{k}^{*}v_{1,k}^{\left(i\right)}+\frac{1}{p}\frac{\tilde{K}-1}{p-1}\sum_{k\in\mathcal{I}_{i}}h_{k}^{*}v_{2,k}-d_{0}U^{\left(i\right)}+\sigma_{0}\dot{W}_{0}^{\left(i\right)},\\
\frac{\mathrm{d}v_{1,k}^{\left(i\right)}}{\mathrm{d}t}= & \:i\left(k^{-1}\beta-kU^{\left(i\right)}\right)v_{1,k}^{\left(i\right)}-U^{\left(i\right)}h_{k}-d_{k}v_{1,k}^{\left(i\right)}+\sigma_{k}\dot{W}_{k}^{\left(i\right)},\quad\left|k\right|\leq K_{1},\\
\frac{\mathrm{d}v_{2,k}}{\mathrm{d}t}= & \:i\left(k^{-1}\beta-kU^{\left(i\right)}\right)v_{2,k}-U^{\left(i\right)}h_{k}-d_{k}v_{2,k}+\sigma_{k}\dot{W}_{k},\qquad k\in\mathcal{I}_{i}.
\end{aligned}
\label{eq:topo_model_batch}
\end{equation}
Notice that we have the fixed-in-time auxiliary variable $h_{k}$
thus we need to take the new quadratic coupling coefficients $c_{kl}=\frac{1}{p}\frac{\tilde{K}-1}{p-1}$
in (\ref{eq:weight_rbm}) between the two different modes $v_{2,k},h_{k}$.
In this case, the $\tilde{K}=K-K_{1}$ small-scale modes $v_{2,k}$
are partitioned into $N_{1}$ small batches, and only the modes in
the $i$-th batch is used to update the $i$-th mean state sample
$U^{\left(i\right)}$. Still, the first $K_{1}$ leading modes $v_{1,k}^{\left(i\right)}$
are resolved explicitly by the ensemble and acting on the corresponding
mean state $U^{\left(i\right)}$ considering their important role
in generating the correct dynamics.

\bibliographystyle{plain}
\nocite{*}
\bibliography{refs}

\end{document}